\documentclass[11pt,reqno]{amsart}
\usepackage{amsmath,amstext,amssymb,amscd}
\usepackage{graphicx} 
\usepackage{amssymb}
\usepackage{changepage}
\usepackage{xcolor}
\usepackage{amsmath}
\usepackage{amsthm}
\usepackage{graphicx}

\newtheorem{theorem}{Theorem}
\newtheorem{lemma}{Lemma}

\usepackage[backend=biber,style=numeric]{biblatex}
\newcommand{\be}{\begin{equation}}
\newcommand{\ee}{\end{equation}}

\newcommand{\ds}{\displaystyle}

\begin{document}
\author[D. Albanez]{Débora A. F. Albanez}
\address{Departamento Acadêmico de Matemática\\ Universidade Tecnológica Federal do Paraná\\ Cornélio Procópio PR, 86300-000, Brazil} \email{deboraalbanez@utfpr.edu.br}

\author[M. Benvenutti]{Maicon José Benvenutti}
\address{Departamento de Matemática\\ Universidade Federal de Santa Catarina\\ Blumenau SC, 89036-004, Brazil} \email{m.benvenutti@ufsc.br}

\author[S. Little]{Samuel Little}
\address{Department of Applied Physics \\Towson University\\
Towson, MD 21252, USA} \email{slittle2@students.towson.edu}

\author[J. Tian]{Jing Tian}
\address{Department of Mathematics \\Towson University\\
Towson, MD 21252, USA} \email{jtian@towson.edu}

\def\bar{\overline}

\title[Parameter Analysis in Data Assimilation for turbulence Models]
{Parameter Analysis in Continuous Data Assimilation for various turbulence Models}

\begin{abstract}
In this study, we conduct parameter estimation analysis on a data assimilation algorithm for two turbulence models: the simplified Bardina model and the Navier-Stokes-$\alpha$ model. Our approach involves creating an approximate solution for the turbulence models by employing an interpolant operator based on the observational data of the systems. The estimation depends on the parameter alpha in the models. Additionally, numerical simulations are presented to validate our theoretical results.
\end{abstract}
\maketitle
\section{Introduction}

Due to the lack of global regularity for the three-dimensional Navier–Stokes equations (NSE), the so-called $\alpha$-models, namely, Navier-Stokes-$\alpha$ \cite{foias2002three}, Leray-$\alpha$ \cite{cheskidov2005leray}, modified Leray-$\alpha$ \cite{ilyin2006modified}, and simplified Bardina model \cite{layton2006well}, i.e., a class of several analytic subgrid-scale models of turbulence in three dimensions, have been studied in the last years as an alternative for abstract mathematical analysis, as well as computational implementation. These models can be seen as a regularization form of the NSE involving a lengthscale parameter $\alpha>0$ and an arbitrary smoothing kernel $G_{\alpha}$ which relates a filtered velocity solution $u$ with an unfiltered velocity $v$ through the filtering relation $u=G_{\alpha}* v$ such that, in some sense $v=v(\alpha)\rightarrow u$ when $\alpha\rightarrow 0^{+}$. For the models cited previously, the kernel $G_{\alpha}$ is taken to be the Green’s function associated
with the Helmholtz operator $I-\alpha^2\Delta$, i.e., $G_{\alpha}=(I-\alpha^2\Delta)^{-1}$. Therefore,
\begin{equation}\label{uandv}
v=u-\alpha^2\Delta u.    
\end{equation}
In this work, we consider two of these $\alpha-$models of turbulence: Firstly, the three-dimensional viscous simplified Bardina turbulence model, under periodic boundary conditions in a domain $\Omega=[0,L]^3$:
\begin{eqnarray}\label{Bardina1}
\left\{
\begin{array}{ll}
v_{t}-\nu\Delta v+(u\cdot\nabla)u+\nabla p=f, \\
\nabla\cdot v=\nabla\cdot u=0&   \\
\end{array}\right.
\end{eqnarray}
where vector $u = (u_1(x,t), u_2(x,t),u_3(x,t))$ is the spatial (filtered) velocity field and $v$ is related with $u$ through \eqref{uandv}, $p =p(x,t)$  is a modified scalar pressure field, $f = f(x,t)$ is a given external force and $\nu>0$ is the kinematic viscosity. This model was considered by Layton and Lewandowski \cite{layton2006well} as being a simpler approximation of
the Reynold stress tensor proposed by Bardina et al. \cite{bardina1980improved}. It is worth mentioning that if $\alpha=0$, from \eqref{uandv} we have $u=v$ and therefore \eqref{Bardina1} reduces to the classical 3D NSE.

The second $\alpha$-model considered is the three-dimensional Navier-Stokes-$\alpha$ equations (NS-$\alpha$), also known as Camassa-Holm equations:

\begin{equation}\label{NSalpha1}
\left\{\begin{array}{l}
v_t-\nu\Delta v
-u\times(\nabla\times v)+\nabla p=f,  \\
\nabla\cdot v=\nabla\cdot u=0
\end{array}\right.
\end{equation}
which is distinguished from \eqref{Bardina1} by the nonlinear term, replaced by $-u\times (\nabla\times v)$. Empirical and numerical experimental data with the solutions of NS-$\alpha$ for mean velocity and Reynolds stresses for turbulent flows in pipes and channels were shown to be successfull  (see \cite{chen1999connection}, \cite{chen1998camassa}, and \cite{chen1999camassa}).

Both the simplified Bardina and Navier-Stokes-$\alpha$ model were explored as a closure approximation for
the Reynolds equations. The Navier-Stokes-$\alpha$ was the first model in the $\alpha$-model family to be studied. It analogues of the Hagen solution, when suitably calibrated, yields good approximations to many experimental data. The simplified Bardina model is consistent with the Navier-Stokes-$\alpha$ in fluid mechanics, in particular, the explicit steady state solution to these two models will match the experimental data. Notably, when comparing the number of degrees of freedom in the long-term dynamics of the solutions, the simplified Bardina model has fewer degrees of freedom than the Navier-Stokes-$\alpha$ \cite{cao2006global}.

Our study investigates the parameter analysis in continuous data assimilation (CDA) for both simplified Bardina model and Navier-Stokes-$\alpha$ model. CDA as explored here was firstly applied to 2D Navier-Stokes equations in \cite{azouani2014continuous} and posteriorly for a wide range of models (see, for instance, \cite{albanez2018continuous}, \cite{albanez2016continuous}, \cite{farhat2015continuous}, \cite{farhat2016data}, \cite{farhat2016charney},
\cite{jolly2017data}, \cite{jolly2015determining}, 
\cite{jolly2019continuous}, \cite{markowich2016continuous}). The concept of CDA involves using $u(t)$, the unknown solution of the system we are interested in, and $I_h(u(t))$, the known physical observations at a spatial resolution size $h$ that are continuous in time $t$, to extract information about $u(t)$. Let us present more details of this procedure below.

Suppose we have the physical system
\begin{align}\label{wmap0}
\frac{du}{dt}=F_{\lambda}(u),
\end{align}
with the initialization point $u_0$ and the parameter $\lambda$ both missing. We want to construct an algorithm for approximate $u(t)$ from the available observational measurements $I_h (u(t))$. Next, we consider the auxiliary assimilated system
\begin{equation}\label{wmap}
\left\{
\begin{array}{l}
\displaystyle \frac{dw}{dt}=F_{\tilde{\lambda}}(w)-\eta I_h (w)+\eta I_h (u),\\ \\
w(0)=w_0,
\end{array}\right.
\end{equation}

where $w_0$ is taken to be arbitrary, $\tilde{\lambda}$ is an a priori guessing to $\lambda$, and $\eta$ is a positive adjustable parameter with dimension $[\frac{1}{T}]$. We integrate the above system with respect to time and obtain the assimilated solution $w$ to approximate the unknown physical solution $u$. In this paper, $\lambda$ is the parameter $\alpha$ related to the filter associated with our system (see (\ref{uandv})), and $I_h$ is a linear interpolant operator that satisfies the following approximation of identity property:
\begin{align}
\|I_h (\varphi) -\varphi\|^{2}_{L^{2}}\leq c_1h^{2}\|\varphi\|^{2}_{L^{2}}+c_2h^{4}\|\Delta \varphi\|^{2}_{L^{2}}, \label{Ih}
\end{align}
where $c_{1}$ and $c_{2}$ are absolute positive constants.
Some studies were developed on the continuous data assimilation algorithm to the NSE and various $\alpha$-models for the case where $\lambda$ is known, that is, $F_{\lambda}=F_{\tilde{\lambda}}$ in (\ref{wmap0})  and (\ref{wmap}) (see for example, \cite{albanez2018continuous}, \cite{azouani2014continuous},   \cite{farhat2017data} and references therein). These results state that for spatial resolution $h$ small enough and the parameter $\mu$ large enough, there is an exponential convergence of the assimilated solution $w$ of (\ref{wmap}) to the physical solution $u$ of (\ref{wmap0}) over time and in suitable norms.
Unlike the above studies, we will focus on the case where $\lambda$ is unknown and  $\tilde{\lambda}$ is an a priori guessing for $\lambda$. We refer to this as the parameter analysis in continuous data assimilation. The motivation for this work stems from the importance of accurately defining parameters that characterize a specific physical scenario when using mathematical models in simulations. Although the underlying physics are often well-understood, identifying the precise parameters for a given model can be challenging. In \cite{carlson2020parameter}, parameter estimation algorithms were applied to 2D Navier-Stokes equations
in the lack of true physical viscosity, where the authors exhibited the large-time error between the true solution of the model and the assimilated solution due to the deviation between the approximate and real physical viscosity. They have also provided an algorithm to recover both the true solution and the true viscosity using only spatially discrete velocity measurements. For a similar approach related to the three-dimensional Brinkman-Forchheimer-extended Darcy model, see \cite{DeboraMaiconBrinkman}.
Furthermore, in \cite{carlsonHudson2022}, the authors developed a nudged system for the classic three-dimensional Lorenz system to create an algorithm for recovering the parameters of this system. They established the convergence of this algorithm with the correct parameters and combined with the computational experiments which proves the efficacy of the algorithm. See  \cite{Biswas_2023},  \cite{yucao2022}, \cite{CHEN2023133552}, \cite{PhysRevFluids.9.054602}, \cite{farhat2020data}, \cite{Martinez_2022}, \cite{Martinez_2024}, \cite{1M1426109}, and references therein for similar results.

Our main results are on the analytical and computational analysis of the long-time error between the solution of the auxiliary assimilated system and the solution of the physical system for both mentioned turbulence models. The length-scale parameter $\alpha>0$ is considered unknown and replaced by a ``guess" parameter $\beta>0$ in the assimilation process. More precisely, for the three-dimensional viscous simplified Bardina model, the assimilated system is given by 
\begin{eqnarray}\label{Bardina1assimilated}
\left\{
\begin{array}{ll}
z_{t}-\nu\Delta z+(w\cdot\nabla)w+\nabla p=f -\eta (I_h (w)-I_h (u)), \\
\nabla\cdot z=\nabla\cdot w=0&   \\
\end{array}\right.
\end{eqnarray}
while for the Navier-Stokes-$\alpha$ equations (NS-$\alpha$), the assimilated system is given by 
\begin{equation}\label{NSalpha1assimilated}
\left\{\begin{array}{l}
z_t-\nu\Delta z
-w\times(\nabla\times z)+\nabla p=f - (\eta-\eta\beta^{2}\Delta) (I_h (w)-I_h (u)),  \\
\nabla\cdot z=\nabla\cdot w=0
\end{array}\right.
\end{equation}
Here, we have the relation
\begin{equation}\label{uandv2}
z=w-\beta^2\Delta w,    
\end{equation}
instead of (\ref{uandv}).

In this work, we prove that under suitable conditions, the approximation solution given in (\ref{Bardina1assimilated}) and in (\ref{NSalpha1assimilated})  converges to the physical solution given in (\ref{Bardina1}) and in (\ref{NSalpha1}), respectively, aside from an error depending the difference between the real and approximating parameters (see Theorems \ref{teoBardina} and \ref{teoNSalpha}). The suitable condition mentioned above is related to the size of resolution parameter $h$ of $I_{h}$ and the parameter $\mu$ compared to the other parameters (see again Theorems \ref{teoBardina} and \ref{teoNSalpha} for the conditions). In (\ref{NSalpha1assimilated}), it is necessary to add the term $\mu\beta^{2}\Delta$ in the assimilation part in order to be able to properly handle the non-linearity $w\times(\nabla\times z)$.
To validate our theoretical analysis, we have presented some numerical simulations. The simulations are powered by a flexible new Python package, Dedalus. Direct numerical simulations with turbulence is very challenging. After trying the traditional finite difference and finite element methods, we decide to use this newly developed package. The advantage of Dedalus is its support for symbolic entry of equations, as well as boundary and initial conditions. Dedalus uses pseudospectral methods to solve differential equations, so as long as the domain is spectrally representable, we can use the package. We perform the comparisons for both the simplified Bardina and Navier-Stokes-alpha models by choosing the nudging parameter $\eta$ within the threshold and outside the threshold. Simulations where $\eta$ lies within the threshold show the convergence, while others outside do not. This serves as a verification of our analysis.


The outline of this work is as follows: In Section 2, we state the mathematical setting,
notations, and classical inequalities used for obtaining all the results. We present the main theoretical results in Section 3. The computational implementation is presented in Section 4. Finally, Conclusions and Appendix are presented at the end.

\section{Preliminaries}
This section presents some mathematical frameworks for the problems we study here. 

\subsection{Mathematical settings}

The statements contained here are standard, and proofs can be found in literature, e.g., in \cite{constantin1988navier},  \cite{Foias_Manley_Rosa_Temam_2001}, \cite{robinson} and \cite{doi:10.1137/1.9781611970050}. Let $\Omega=[0,L]^3$, we denote by $L^{p} $   the usual three-dimensional Lebesgue vector spaces. For each $s \in \mathbb{R}$, we define the Hilbert space
\begin{equation}
\dot{H}_{s}= \left\{u(x)=\sum_{K \in \mathbb{Z}^{3}\backslash\{0\}}\hat{u}_{K}e^{2\pi i \frac{K\cdot x}{L}};  \,\, \hat{u}_{K}= \overline{\hat{u}_{-K}}, \,\, \sum_{K \in \mathbb{Z}^{3}\backslash\{0\}} |K|^{2s}|\hat{u}_{K}|^{2}<\infty\right\}, \nonumber
\end{equation}
with the inner product
\begin{equation}
(u,v)_{\dot{H}_{s}}= L^{3}\sum_{K \in \mathbb{Z}^{3}\backslash\{0\}} \left(\frac{2\pi|K|}{L}\right)^{2s} \hat{u}_{K}\cdot\overline{\hat{v}_{K}}, \nonumber
\end{equation}
and closed subspace $\dot{V}_{s}$ defined as
\begin{equation}
\left\{u(x)=\sum_{K \in \mathbb{Z}^{3}\backslash\{0\}}\hat{u}_{K}e^{2\pi i \frac{K\cdot x}{L}};  \,\, \hat{u}_{K}= \overline{\hat{u}_{-K}}, \,\, \hat{u}_{K}\cdot K=0\right\}, \nonumber
\end{equation}
where $\sum_{K \in \mathbb{Z}^{3}\backslash\{0\}} |K|^{2s}|\hat{u}_{K}|^{2}<\infty.$
We have that $\dot{V}_{s_{1}} \subset \dot{V}_{s_{2}}$ if $s_{1} \geq s_{2}$ and $\dot{V}_{-s}$ is the dual of $\dot{V}_{s}$, for all $s \geq 0$.

We denote by $\mathcal{P}:\dot{H}_{s}\rightarrow \dot{V}_{s}$ the classical Helmholtz-Leray orthogonal projection given by 
\begin{equation}
 \mathcal{P}u= \sum_{K \in \mathbb{Z}^{3}\backslash\{0\}} \left(\hat{u}_{K}-\frac{K(\hat{u}_{K} \cdot K)}{|K|^{2}} \right)e^{2\pi i \frac{K\cdot x}{L}} \nonumber 
\end{equation}
 and $A: \dot{V}_{2s} \rightarrow \dot{V}_{2s-2}$ the operator given by 
\begin{equation}
Au=\sum_{K \in \mathbb{Z}^{3}\backslash\{0\}} \frac{4\pi^{2}|K|^{2}}{L^{2}}\hat{u}_{K}e^{2\pi i \frac{K\cdot x}{L}}. \nonumber
\end{equation}
Moreover, we have $Au=-\Delta u = -\mathcal{P} \Delta u=-\Delta \mathcal{P} u$.

We adopt the classical notations  $H= \dot{V}_{0}$, $V= \dot{V}_{1}$, $D(A)= \dot{V}_{1}$, $V^{'}= \dot{V}_{-1}$, $D^{'}= \dot{V}_{-2}$, $\|u\|=\|u\|_{L^{2}}$ and $(u,v)=(u,v)_{L^{2}}$, 
 with the identities
\begin{align}
\|u\|_{H}= \|u\|, \,\,\,\,\, \|u\|_{V}= \| \nabla u\| \,\,\,\,\, \mbox{and} \,\,\,\,\,\|u\|_{D(A)}= \|A u\|, \nonumber
\end{align}
and the  Poincar{\'e} inequalities
 \begin{equation}
 \|u\|^{2}\leq\lambda^{-1}_{1}\|\nabla u\|^{2} \,\,\mbox{  for all } u\in V \mbox{  and  }\,\, \|\nabla u\|^{2}\leq\lambda^{-1}_{1}\|Au\|^{2}\,\, \mbox{ for all  } u\in D(A), \label{Poincare}
\end{equation}
where $\displaystyle \lambda_{1}= \frac{4\pi^{2}}{L^{2}}$.  We also recall  some particular cases of the Gagliardo-Nirenberg inequality:
\begin{equation}
\left\{
\begin{array}{ll}
\| g\|_{L^{3}} \leq c\|g\|^{\frac{1}{2}}\| \nabla g\|^{\frac{1}{2}}, & \forall\, g\in V; \\
\| g\| _{L^{4}} \leq c\|g\|^{\frac{1}{4}}\| \nabla g\|^{\frac{3}{4}}, & \forall\, g\in V; \\
\| g\| _{L^{6}}\leq c\|\nabla  g\|, & \forall\, g\in V,
\end{array}\right.  \label{Gagliardo-Nirenberg}
\end{equation}%
where  $c$  is a dimensionless constant.

 Finally, for each $\alpha>0$, we have
\begin{eqnarray}
(I+\alpha^{2}A)^{-1}u = \sum_{K \in \mathbb{Z}^{3}\backslash\{0\}} \frac{L^{2}}{L^{2}+4\alpha^{2}\pi^{2} |K|^{2}}\hat{u}_{K}e^{2\pi i \frac{K\cdot x}{L}}, \nonumber
\end{eqnarray}
with the estimates
\begin{align}
\|(I+\alpha^{2}A)^{-1}u\| \leq \|u\| \,\,\,\, \mbox{ and } \,\,\,\, \|(I+\alpha^{2}A)^{-1}u\| \leq \frac{1}{\alpha^{2}}\|u\|_{D^{'}}. \label{inequality-1}
\end{align}

\subsection{Abstract formulations}

Define the bilinear extended operators $B,\Tilde{B}:V\times V\rightarrow V^{'}$ for every $u,v\in V$ (see \cite{foias2002three}):
$$B(u,v)=\mathcal{P}[(u\cdot\nabla)v]\hspace{0.5cm}\mbox{and}\hspace{0.5cm}\Tilde{B}(u,v)=-\mathcal{P}(u\times(\nabla\times v))$$
The bilinear operators have the following properties for all $u,v,w\in V$:
\begin{equation*}
 \langle B(u,v),w\rangle_{V'}=-\langle B(u,w),v\rangle_{V'}    
\end{equation*}
\begin{equation*}
   \langle \Tilde{B}(u,v),w\rangle_{V'}=-\langle \Tilde{B}(w,v),u\rangle_{V'}
\end{equation*}
where $\langle \,\cdot\,,\,\cdot\,\rangle$ is the duality in the proper situation, which implies
\begin{equation}\label{zero}
 \langle B(u,v),v\rangle_{V'}=0\hspace{0.5cm}\mbox{and}\hspace{0.5cm}\langle \Tilde{B}(u,v),u\rangle_{V'}=0.   
\end{equation}

From inequalities given in (\ref{Poincare}), (\ref{Gagliardo-Nirenberg}), and (\ref{inequality-1}) and duality, we obtain 
\begin{equation}\label{bilin6}
\|(I+\alpha^{2}A)^{-1}\tilde{B}(u,v)\|\leq c^{2}\frac{\lambda_{1}^{-\frac{1}{4}}}{\alpha^{2}}\|\nabla u\|\|v\|.
\end{equation}

Using Leray projector and assuming the forcing term $f\in L^{\infty}(0,T;H)$, we rewrite the system \eqref{Bardina1} as 
\begin{eqnarray}\label{Bardina1Leray}
\left\{
\begin{array}{ll}
\ds\frac{dv}{dt}+\nu Av+B(u,u)=f, \\
v=u+\alpha^2Au,\\
\nabla\cdot v=\nabla\cdot u=0,\,\,\, u(0)=u_0.&   \\
\end{array}\right.
\end{eqnarray}

and the system \eqref{NSalpha1} as 

\begin{equation}\label{NSalpha1Leray}
\left\{\begin{array}{l}
\ds\frac{dv}{dt}+\nu Av+\Tilde{B}(u,v)=f,\\
v=u+\alpha^2Au,\\
\nabla\cdot v=\nabla\cdot u=0, \,\,\, u(0)=u_0.
\end{array}\right.
\end{equation}

Similarly, we rewrite the system \eqref{Bardina1assimilated} as 
\begin{eqnarray}\label{Bardina1LerayCDA}
\left\{
\begin{array}{ll}
\ds\frac{dz}{dt}+\nu Av+B(w,w)=f-\eta\mathcal{P}(I_hw-I_hu), \\
z=w+\beta^2Aw,\,\,\, w(0)=w_0.\\
\end{array}\right.
\end{eqnarray}

and the system \eqref{NSalpha1assimilated} as 

\begin{equation}\label{NSalpha1LerayCDA}
\left\{\begin{array}{l}
\ds\frac{dz}{dt}+\nu Az+\Tilde{B}(w,z)=f-\eta(I+\beta^2A)(I_hw-I_hu),\\
z=w+\beta^2Aw,\,\,\, w(0)=w_0.\\
\end{array}\right.
\end{equation}
with $I_h$ a linear operator satisfying the property \eqref{Ih}.

Results concerning the global well-posedness of the four above systems with initial condition $u(0)$, $w(0)\in V$ can be found in \cite{cao2006global},  \cite{foias2002three}, \cite{albanez2018continuous}, and  \cite{albanez2016continuous}, respectively. To all these systems, the solutions have the properties $u$, $w \in C([0,T);V)\cap L^2([0,T];D(A))$ with $\frac{du}{dt}$, $\frac{dw}{dt}\in L^2([0,T);H)$.

Finally, to establish our primary theoretical outcomes, we require the following version of the Gronwall's inequality, whose proof is in the Appendix.

\begin{lemma}[Gronwall Inequality]\label{Gronwall}
Let $\xi:[t_0,\infty)\rightarrow[0,\infty)$ absolute continuous function and $\beta:[t_0,\infty)\rightarrow[0,\infty)$ locally integrable. Suppose that there exist positive constants $C,M$ and $T$ such that
\begin{equation}\label{int1}
    \sup\limits_{t\geq t_0}\ds\int_t^{t+T}\beta(s)\,ds\leq M
\end{equation}

\begin{equation}\label{ineq1}
    \ds\frac{d}{dt}\xi(t)+C\xi(t)\leq \beta(t)
\end{equation}
  is satisfied for all $t\geq t_0$. Then
  \begin{equation}\label{int211}
    \xi(t)\leq e^{-C(t-t_0)}\xi(t_0)+M\left(\frac{e^{2CT}}{e^{CT}-1}\right)  
  \end{equation}
  
   for all $t\geq t_0$.
\end{lemma}

\subsection{Estimates for reference solutions}
We state here the inequalities related to global unique solutions of systems \eqref{Bardina1Leray} and \eqref{NSalpha1Leray} that are useful to determine the long-time error exhibited in Theorems \ref{teoBardina} and \ref{teoNSalpha}.\\
Firstly, Lemma \ref{lem1} below refers to both solutions of the Bardina \eqref{Bardina1Leray} and Navier-Stokes-$\alpha$ \eqref{NSalpha1Leray}:
\begin{lemma}\label{lem1}
For all $t\geq 0$, the global unique solution of \eqref{Bardina1Leray} (or \eqref{NSalpha1Leray}) satisfies
$$\|u(t)\|^2+\alpha^2\|\nabla u(t)\|^2\leq e^{-\nu\lambda_1t}(\|u(0)\|^2+\alpha^2\|\nabla u(0)\|^2)+\ds\frac{1}{\lambda_1^2\nu^2}\cdot \sup\limits_{s\geq 0}\|f(s)\|^2.$$
In particular, if 
\begin{equation}\label{M1}  M_1:=\|u(0)\|^2+\alpha^2\|\nabla u(0)\|^2+\frac{1}{\lambda_1^2\nu^2}\cdot \sup\limits_{s\geq 0}\|f(s)\|^2,
\end{equation} 
we have for all $t\geq 0$,
$$\|u(t)\|^2+\alpha^2\|\nabla u(t)\|^2\leq M_1.$$
Moreover, for $T>0$, we have
\begin{equation}\label{inttT}
    \ds\int_t^{t+T}\|\nabla u(s)\|^2+\alpha^2\|Au(s)\|^2\,ds\leq \ds\frac{M_1}{\nu}+\ds\frac{T}{\nu^2\lambda_1}\sup\limits_{s\geq 0}\|f(s)\|^2.
\end{equation}

\end{lemma} 
\begin{proof} Multiplying the system \eqref{Bardina1Leray} (or system \eqref{NSalpha1Leray}) by the strong solution $u(t)$, integrating in $\Omega$ and integrating by parts, we get
\begin{eqnarray}\label{first}
    &\ds\frac{1}{2}\ds\frac{d}{dt}(\|u(t)\|^2+\alpha^2\|\nabla u(t)\|^2)+\nu(\|\nabla u(t)\|^2+\alpha^2\|Au(t)\|^2)=(f(t),u(t))_{L^2}\nonumber\\
    &\leq \ds\frac{1}{2\nu\lambda_1}\|f(t)\|^2+\ds\frac{\nu\lambda_1}{2}\|u(t)\|^2\leq \ds\frac{1}{2\nu\lambda_1}\|f(t)\|^2+\ds\frac{\nu}{2}\|\nabla u(t)\|^2,
\end{eqnarray}
and thus
\begin{equation}\label{above}
\ds\frac{d}{dt}(\|u(t)\|^2+\alpha^2\|\nabla u(t)\|^2)+\nu\lambda_1(\|u(t)\|^2+\alpha^2\|\nabla u(t)\|^2)\leq\ds\frac{1}{\nu\lambda_1}\|f(t)\|^2.    \nonumber
\end{equation}

By classical Gronwall's inequality (see \cite{evans2022partial}), we get
\begin{eqnarray}
    &\|u(t)\|^2+\alpha^2\|\nabla u(t)\|^2\\
    &\leq e^{-\nu\lambda_1t}(\|u(0)\|^2+\alpha^2\|\nabla u(0)\|^2)+\ds\frac{1}{\nu\lambda_1}\ds\int_0^te^{\nu\lambda_1(s-t)}\|f(s)\|^2\,ds\nonumber\\
    &\leq e^{-\nu\lambda_1t}(\|u(0)\|^2+\alpha^2\|\nabla u(0)\|^2)+\ds\frac{1}{\nu^2\lambda_1^2}\cdot\sup\limits_{s\geq 0}\|f(s)\|^2.\nonumber
\end{eqnarray}
We obtain \eqref{inttT} from integrating  \eqref{first}.
\end{proof}

Now, Lemma \ref{lem3} refers to global unique solution of the Bardina system \eqref{Bardina1Leray}: 
\begin{lemma}\label{lem3}
    For all $T>0$ and $t\geq 0$, we have the following estimate for the solution of \eqref{Bardina1Leray}:
    $$\ds\int_t^{t+T}\|u_t(s)\|^2\,ds\leq \frac{6M_1\nu}{\alpha^2}+\frac{81}{256}\frac{c^{16}M_1^5T}{\nu^6\alpha^8}+
    3T\left(1+\frac{2}{\alpha^2\lambda_1}\right)\sup\limits_{s\geq 0}\|f(s)\|^2$$
\end{lemma}
where $c$ is given in  (\ref{Gagliardo-Nirenberg}) and $M_{1}$ in (\ref{M1}).

\begin{proof}
     Applying the inverse operator $(I+\alpha^2A)^{-1}$ in \eqref{Bardina1Leray}, we obtain
$$u_t+\nu Au+(I+\alpha^2)^{-1}B(u,u)=(I+\alpha^2A)^{-1}f.$$
Then,
\begin{eqnarray}
   \|u_t\|\leq & \nu\|Au\|+\|(I+\alpha^2A)^{-1}B(u,u)\|+\|(I+\alpha^2A)^{-1}f\|\nonumber\\
    \leq & \hspace{-3.8cm} \nu\|Au\|+\|B(u,u)\|+\|f\|\nonumber\\
    \leq & \nu\|Au\|+c^2\|u\|^{1/4}\|\nabla u\|^{3/4}\|\nabla u\|^{1/4}\|Au\|^{3/4}+\|f\|\nonumber\\
    = &\hspace{-1.8cm} \nu \|Au\|+c^2\|u\|^{1/4}\|\nabla u\|\, \|Au\|^{3/4}+\|f\|\nonumber\\
    \leq &\hspace{-3cm} 2\nu\|Au\|+\frac{27}{256}\frac{c^{8}\|u\|\,\|\nabla u\|^4}{\nu^3}+\|f\|.\nonumber
\end{eqnarray}
Therefore 
$$\|u_t\|^2\leq 6\nu^2\|Au\|^2+\frac{81}{256}\frac{c^{16}\|u\|^2\,\|\nabla u\|^8}{\nu^6}+3\|f\|^2.$$
From Lemma \ref{lem1}, we get
\begin{eqnarray}
   \ds\int_t^{t+T}\|u_t(s)\|^2ds\leq & \hspace{-2cm}6\nu^2\left[\ds\frac{M_1}{\nu\alpha^2}+\ds\frac{T}{\alpha^2\nu^2\lambda_1}\sup\limits_{s\geq 0}\|f(s)\|^2\right]\nonumber\\
    + & \ds\frac{81}{256}\ds\frac{c^{16}}{\nu^6}\left[M_1\cdot\left(\ds\frac{M_1}{\alpha^2}\right)^4\right]T+3T\sup\limits_{s\geq 0}\|f(s)\|^2. \nonumber
\end{eqnarray}
\end{proof}

Finally, we give an estimate for the global unique solution of Navier-Stokes-$\alpha$ system \eqref{NSalpha1Leray} in Lemma \ref{lem7}.

\begin{lemma}\label{lem7}
    For the time-derivative $u_t$ of the solution to the system \eqref{NSalpha1Leray}, we have the following estimate:
\begin{eqnarray}
\ds\int_t^{t+T}\|u_{t}(s)\|^{2}ds &\leq &
\ds\frac{6}{\alpha^2}\left(\frac{M_1}{\nu}+\frac{T}{\nu^2\lambda_1}\sup\limits_{s\geq 0}\|f(s)\|^2\right)\left(\nu^2+\ds\frac{c^2M_1}{16\pi^4\lambda_1^{1/2}\alpha^2}\right)\nonumber\\
&+&\ds\frac{3c^2M_1^2}{16\pi^4\lambda_1^{1/2}\alpha^6}\cdot T+3T\sup\limits_{s\geq 0}\|f(s)\|^2,\label{u_tNSalpha}
\end{eqnarray}
where $c$ is given in  (\ref{Gagliardo-Nirenberg}) and $M_{1}$ in (\ref{M1}). 
\end{lemma}
\begin{proof} Applying the operator $(I+\alpha^2A)^{-1}$ in \eqref{NSalpha1Leray}, we have the following equivalent Navier-Stokes-$\alpha$ equations:
$$
\frac{d u}{d t}+\nu A u+(I+\alpha^{2} A)^{-1}\tilde{B}(u, u+\alpha^{2} A u)=(I+\alpha^{2} A)^{-1} f.
$$
Then 

\begin{eqnarray}
 \|u_{t}\|_{L^{2}} & \leq \nu\|A u\|+\|(I+\alpha^{2}A)^{-1} f\|+\|(I+\alpha^{2} A)^{-1}\tilde{B}(u, u+\alpha^{2} Au)\| \nonumber\\
&\hspace{-2cm} \leq \nu\|Au\|+\|f\|+\ds\frac{1}{4\pi^{2}\alpha^{2}}\|\tilde{B}(u, u+\alpha^{2} Au)\|_{H^{-2}}\nonumber \\
&\hspace{-2cm} \leq \nu\|A u\|+\ds\frac{c}{4\pi^2\alpha^2\lambda_{1}^{1/4}}\|\nabla 
 u\|\,\|u+\alpha^{2} Au\|+\|f\|\nonumber \\
&\hspace{-1cm} \leq \nu\|A u\|+\ds\frac{c}{4\pi^2\alpha^2\lambda_{1}^{1/4}}\|\nabla u\| (\| u \| + \alpha^2 \| Au \| )+\|f\|. \nonumber  
\end{eqnarray}

Squaring both sides above and coupling the terms, we obtain
 $$   \|u_t\|^2 \leq  6\|Au\|^2\left(\nu^2+\frac{c^2}{16\pi^4\lambda_1^{1/2}}\|\nabla u\|^2\right)
    +\ds\frac{3c^2}{16\pi^4\alpha^4\lambda_1^{1/2}}\|\nabla u\|^2\|u\|^2+3\|f\|^2$$
   Finally, integrating the inequality above over $[t,t+T]$ and using Lemma \ref{lem1}, we obtain \eqref{u_tNSalpha}.
\end{proof}

\section{Results on the long-time error}

In this section, we present our analytical results, demonstrating that the real-state solution $u(t)$ of both models (Bardina and Navier-Stokes-$\alpha$) can be approximated by ``alternative" solutions $w(t)$, which use an approximate parameter $\beta > 0$ instead of the original $\alpha > 0$. Specifically, we show that the $L^2$-error between these solutions decays exponentially over time, with the exception of an error that is controlled by the difference between the parameters $\beta$ and $\alpha$.\\

Theorems \ref{teoBardina} refers to Bardina system, while Theorem \ref{teoNSalpha} is related to Navier-Stokes-$\alpha$ equations.

\begin{theorem}\label{teoBardina}
 Let $u_0,w_0\in V$ be initial data of the systems \eqref{Bardina1Leray} and \eqref{Bardina1LerayCDA} respectively, and $u$ and $w$ be the respective global solutions. Suppose the parameters $\alpha,\beta>0$,\, $\eta$ large enough and $h$ small enough such that the conditions below are satisfied:
    \begin{enumerate}
    \item $\eta> \ds\frac{27c^4}{16}\left(\frac{3}{\nu}\right)^3\ds\frac{M_1^2}{\alpha^4},$
        \item $2\eta c_1 h^2<\nu$ and $2\eta c_2h^4<\nu\beta^2$;
         
           \end{enumerate}
where $c$ is given in (\ref{Gagliardo-Nirenberg}), $M_{1}$ in (\ref{M1}) and $c_1$ and $c_2>0$ in \eqref{Ih}.  Then, the following inequality for the difference between physical and assimilated solutions, that is, $g(t):=w(t)-u(t)$, is valid  for all $t \geq 0$:
\begin{align}\label{error}
    \|g(t)\|^2 + \beta^2\|\nabla g(t)\|^2 \leq & e^{-\frac{\lambda_1\nu}{2}t}(\|g(0)\|^2+\beta^2\|\nabla g(0)\|^2)
 \nonumber \\&+ \ds\frac{2e}{e^{1/2}-1}\cdot \ds\frac{|\beta^2-\alpha^2|^2}{\beta^2}M_{2},
\end{align}
where
 \begin{align}
   M_{2}=\ds\frac{7M_1}{\alpha^2}+\frac{81}{256}\frac{c^{16}M_1^{5}}{\nu^8\alpha^2\lambda_1}+\left(\ds\frac{1}{\nu^2\lambda_1}+\frac{3}{\nu^2\lambda_1^2\alpha^2}\right)\sup\limits_{s\geq 0}\|f(s)\|^2.
\end{align}
\end{theorem}

\begin{proof} Subtracting \eqref{Bardina1Leray} from \eqref{Bardina1LerayCDA} yields
\begin{align}\label{eqdifference}
\frac{d}{dt}(g+\beta^2Ag+(\beta^2-\alpha^2)Au)&+\nu A(g+\beta^2Ag+(\beta^2-\alpha^2)Au)+B(g,u)\nonumber \\ &+B(w,g)=-\eta \mathcal{P}I_hg, 
\end{align}
with $\mbox{div}\, g=0$. Taking the $D^{'}$-dual action  with $g$ in \eqref{eqdifference}, we get
\begin{align}
    \frac{1}{2}\frac{d}{dt}(\|g\|^2+\beta^2\|\nabla g\|^2)&+\nu(\|\nabla g\|^2+\beta^2\|Ag\|^2) = (\alpha^2-\beta^2)\langle\frac{d}{dt}Au,g \rangle\nonumber\\
    &- \nu(\beta^2-\alpha^2)(Au,Ag)
   -(B(g,u),g)\nonumber \\&-(B(w,g),g)-\eta(I_hg,g).
    \end{align}

Using \eqref{zero} and general Hölder's inequality, we obtain
\begin{eqnarray*}
    \frac{1}{2}\frac{d}{dt}(\|g\|^2+&\hspace{-2cm}\beta^2\|\nabla g\|^2)+\nu(\|g\|^2+\beta^2\|Ag\|^2) \leq |\alpha^2-\beta^2|\|u_t\|\,\|Ag\| \nonumber\\
     &+\nu(\beta^2-\alpha^2)\|Au\|\,\|Ag\|+c\|g\|^2_{L^4}\|\nabla u\|+\eta\|I_hg-g\|\,\|g\|-\eta\|g\|^2.\nonumber\\
    \end{eqnarray*}
We estimate each term of right-hand side above, firstly in order to get some of them absorbed by dissipation term:
$$|\alpha^2-\beta^2|\|u_t\|\,\|Ag\|\leq\ds\frac{\nu\beta^2}{4}\|Ag\|^2+\ds\frac{|\alpha^2-\beta^2|^2}{\nu\beta^2}\|u_t\|^2; $$

$$\nu(\beta^2-\alpha^2)\|Au\|\,\|Ag\|\leq \ds\frac{\nu\beta^2}{4}\|Ag\|^2+\ds\frac{\nu}{\beta^2}\|Au\|^2|\beta^2-\alpha^2|^2;$$
$$c\|g\|^2_{L^4}\|\nabla u\|\leq c\|g\|^{1/2}\|\nabla g\|^{3/2}\|\nabla u\|\leq \ds\frac{\nu}{2}\|\nabla g\|^2+\frac{c^4}{32}\left(\frac{3}{\nu}\right)^3\|g\|^2\|\nabla u\|^4;$$
$$ \eta\|I_hg-g\|\,\|g\|\leq \frac{\eta}{2}\|I_hg-g\|^2+\frac{\eta}{2}\|g\|^2\leq\ds\frac{\eta c_1h^2}{2}\|\nabla g\|^2+\frac{\eta c_2h^4}{2}\|Ag\|^2+\frac{\eta}{2}\|g\|^2.$$

Therefore
\begin{eqnarray}\label{chk}
&\ds\frac{d}{dt}(\|g\|^2+\beta^2\|\nabla g\|^2)  +(\nu-\eta c_1h^2)\|\nabla g\|^2+ (\beta^2\nu-\eta c_2 h^4)\|Ag\|^2 \nonumber\\
&\leq \left[\ds\frac{c^4}{16}\left(\frac{3}{\nu}\right)^3\|\nabla u\|^4-\eta\right]\|g\|^2+2\left(\ds\frac{1}{\nu\beta^2}\|u_t\|^2+\ds\frac{\nu}{\beta^2}\|Au\|^2\right)|\beta^2-\alpha^2|^2.\nonumber\\
\end{eqnarray}
By Lemma \ref{lem1}, we have that 
$$\|\nabla u\|^4\leq\ds\frac{M_1^2}{\alpha^4}, $$
and therefore, using conditions (1) and (2) required in the statement of Theorem, we conclude that
\begin{align*}
&\ds\frac{d}{dt}(\|g\|^2 +\beta^2\|\nabla g\|^2)+\ds\frac{\lambda_1\nu}{2}(\|g\|^2 +\beta^2\|\nabla g\|^2)\\
&\leq 2\left(\ds\frac{1}{\nu\beta^2}\|u_t\|^2+\ds\frac{\nu}{\beta^2}\|Au\|^2\right)|\beta^2-\alpha^2|^2.
\end{align*}
Denoting 
$$\xi(t)=\|g(t)\|^2+\beta^2\|\nabla g(t)\|^2,\,\,\,C=\ds\frac{\lambda_1\nu}{2},$$  $$\gamma(t)=2\ds\frac{|\beta^2-\alpha^2|^2}{\beta^2}\left(\ds\frac{1}{\nu}\|u_t(t)\|^2+\nu\|Au(t)\|^2\right),$$ 
we have that 
\begin{equation}\label{grn}
\ds\frac{d}{dt}\xi(t)+C\xi(t)\leq \gamma(t).    
\end{equation}

From Lemmas \ref{lem1} and \ref{lem3}, we have for all $T\geq t$,
\begin{eqnarray*}
&\int_t^{t+T}\gamma(s)ds\\
& = \ds\frac{2|\beta^2-\alpha^2|^2}{\nu\beta^2}\ds\int_t^{t+T}\|u_t(s)\|^2ds+2\nu\frac{|\beta^2-\alpha^2|^2}{\beta^2}\ds\int_t^{t+T}\|Au(s)\|^2ds\nonumber\\
& \leq   \ds\frac{2|\beta^2-\alpha^2|^2}{\nu\beta^2} \left[\ds\frac{6M_1\nu}{\alpha^2}+\frac{81}{256}\frac{c^{16}M_1^5T}{\nu^6\alpha^8}+3T\left(1+\ds\frac{2}{\alpha^2\lambda_1}\right)\sup\limits_{s\geq 0}\|f(s)\|^2\right]\nonumber\\
& \hspace{-2cm}+ 2\nu\ds\frac{|\beta^2-\alpha^2|^2}{\alpha^2\beta^2}\left(\ds\frac{M_1}{\nu}+\frac{T}{\nu^2\lambda_1}\sup\limits_{s\geq 0}\|f(s)\|^2\right),\nonumber
\end{eqnarray*}
Choosing $T=(\lambda_1\nu)^{-1}$, we obtain
\begin{eqnarray*}
&\int_t^{t+\frac{1}{\lambda_1\nu}}\gamma(s)\,ds\\
&\leq   \ds\frac{2|\beta^2-\alpha^2|^2}{\beta^2} \bigg[\ds\frac{6M_1}{\alpha^2}+\frac{81}{256}\frac{c^{16}M_1}{\nu^8\alpha^8\lambda_1}+\left(\ds\frac{1}{\nu^2\lambda_1}+\ds\frac{2}{\nu^2\alpha^2\lambda_1^2}\right)\sup\limits_{s\geq 0}\|f(s)\|^2\nonumber\\
& + \ds\frac{M_1}{\alpha^2}+\ds\frac{1}{\nu^2\lambda_1^2\alpha^2}\sup\limits_{s\geq 0}\|f(s)\|^2\bigg]\nonumber\\
& = \ds\frac{2|\beta^2-\alpha^2|^2}{\beta^2}\bigg[\ds\frac{7M_1}{\alpha^2}+\frac{81}{256}\frac{c^{16}M_1}{\nu^8\alpha^8\lambda_1} +\left(\ds\frac{1}{\nu^2\lambda_1}+\ds\frac{3}{\nu^2\alpha^2\lambda_1^2}\right)\sup\limits_{s\geq 0}\|f(s)\|^2\bigg].\nonumber
\end{eqnarray*}
Using Lemma \ref{Gronwall} applied to \eqref{grn}, with  $$M=\ds\frac{2|\beta^2-\alpha^2|^2}{\beta^2}\bigg[\ds\frac{7M_1}{\alpha^2}+\frac{81}{256}\frac{c^{16}M_1}{\nu^8\alpha^8\lambda_1} +\left(\ds\frac{1}{\nu^2\lambda_1}+\ds\frac{3}{\nu^2\alpha^2\lambda_1^2}\right)\sup\limits_{s\geq 0}\|f(s)\|^2\bigg], $$

we finally conclude the estimated error \eqref{error}.

\end{proof}

\begin{theorem}\label{teoNSalpha}
Let $u_0,w_0\in V$ and $u$ and $w$ be solutions to the systems \eqref{NSalpha1Leray} and \eqref{NSalpha1LerayCDA}, respectively, with initial data $u_0$ and $w_0$. Assume that the following conditions are satisfied:
\begin{enumerate}
\item $\eta> \frac{576c^{8}M_{1}^{2}}{\alpha^{4} \nu^{3}}$; 
\item $\eta c_{1}h^{2}+\ds\frac{\eta^{2}c_{1}\beta^{2}h^{2}}{\nu}+\frac{432c^{8}\beta^{2}M^{2}_{1}}{\nu^{3}\alpha^{4}}-\eta\beta^{2}<\frac{11\nu}{16};$
\item $\eta c_{2}h^{4}+\ds\frac{\eta^{2}\beta^{2}c_{2}h^{4}}{\nu}< \frac{\nu}{4}\beta^{2}.$ 
 \end{enumerate}
where $c$ is given in (\ref{Gagliardo-Nirenberg}), $M_{1}$ in (\ref{M1}) and $c_1$ and $c_2>0$ in \eqref{Ih}. Then, the following inequality for the difference between physical and assimilated solutions, that is, $g(t):=w(t)-u(t)$, is valid for all  $t \geq 0$:
\begin{align}
\|g(t)\|^{2}+\beta^{2}\|\nabla g(t)\|^{2}  \le & e^{-\frac{\lambda_{1} \nu}{2} t}\left(\|g(0)\|^{2}+\beta^{2}\|\nabla g(0)\|^{2}\right) \nonumber\\
& +\left(\frac{2e}{e^{1 / 2}-1}\right)\frac{|\beta^2-\alpha^2|^2}{\beta^2}M_{3}, \label{teo22}
\end{align}
where
\begin{align}
M_{3} &  := \ds\frac{4}{\nu}\Big[\frac{6}{\alpha^2}\Big(\nu^2+\frac{c^4M_1}{\lambda_1^{1/2}\alpha^2}\Big)\Big(\frac{M_1}{\nu}+\frac{1}{\nu^3\lambda_1^2}\sup\limits_{s\geq 0}\|f(s)\|^2 \Big)\nonumber\\
& +\frac{3c^4M_1^2}{\alpha^6\nu\lambda_1^{3/2}}+\frac{3}{\nu\lambda_1}\sup\limits_{s\geq 0}\|f(s)\|^2\Big]+\frac{4\nu}{\alpha^2}\Big(\frac{M_1}{\nu}+\frac{1}{\nu^3\lambda_1^2}\sup\limits_{s\geq 0}\|f(s)\|^2 \Big)\nonumber\\
& + \ds\frac{2c^4M_1}{\nu\alpha^4 \lambda^{\frac{1}{2}_{1}}}\Big(\frac{M_1}{\nu}+\frac{1}{\nu^3\lambda_1^2}\sup\limits_{s\geq 0}\|f(s)\|^2 \Big).\label{MMM}
\end{align}
\end{theorem}
\begin{proof}
Taking the difference \eqref{NSalpha1LerayCDA}$-$ \eqref{NSalpha1Leray}, we get 
$$\begin{array}{rcl}
 \ds\frac{\mathrm{d} }{\mathrm{d} t} (g+\beta^{2}Aw-\alpha^{2}Au)&\hspace{-0.3cm}+&\hspace{-0.3cm}\nu\left[Ag+A(\beta^{2}Aw-\alpha^{2}Au)\right]\\
 &\hspace{-0.3cm}+&\hspace{-0.3cm}\tilde{B}(w,z)-\tilde{B}(u,v)=-\eta\mathcal{P}(I+\beta^{2}A)I_{h}g, 
\end{array}$$

with $\mbox{div}z,g=0$. Rewriting the following above terms:

$$g+\beta^{2}Aw-\alpha^{2}Au=g+\beta^{2}Ag+(\beta^{2}-\alpha^{2})Au;$$
$$Ag+A(\beta^{2}Aw-\alpha^{2}Au)=Ag+A\left [\beta^{2}Ag+(\beta^{2}-\alpha^{2})Au\right ],$$

yields
$$\begin{array}{rcl}
  \ds\frac{\mathrm{d} }{\mathrm{d} t} (g+\beta^{2}Ag+(\beta^{2}-\alpha^{2})Au)   &\hspace{-0.3cm} + &\hspace{-0.3cm} \nu A\left[g+\beta^{2}Ag+(\beta^{2}-\alpha^{2})Au\right]\\
     &\hspace{-0.3cm} + & \hspace{-0.3cm} \tilde{B}(w,z)-\tilde{B}(u,v)=-\eta P I_{h}g-\eta \beta^{2}A I_{h}g,
     \end{array}$$

Taking the duality $\left \langle\cdot ,g \right \rangle _{D'} $,\;we get

\begin{eqnarray}
 \frac{1}{2}\frac{\mathrm{d} }{\mathrm{d} t}  (\left \| g \right \| ^{2}+\beta^{2}\left \| \nabla g \right \| ^{2})&+&\hspace{-0.2cm}(\beta^{2}-\alpha^{2})\left \langle \frac{\mathrm{d} }{\mathrm{d} t} Au,g \right \rangle_{D'} +\nu ( \| \nabla g \| ^{2}+\beta^{2} \|  Ag  \| ^{2}) 
 \nonumber\\
 &+&\hspace{-0.2cm} \nu (\beta^{2}-\alpha^{2})(Au,Ag)+(\tilde{B}(w,z)-\tilde{B}(u,v) ,g)\nonumber\\
 &=&\hspace{-0.2cm}-\eta(I_{h}g,g)-\eta \beta^{2}\left \langle A I_{h}g,g \right \rangle _{D(A)'}.  \nonumber    
\end{eqnarray}
 Writing $z-v=g+\beta^{2}Ag+(\beta^{2}-\alpha^{2})Au$, we rewrite the nonlinear term as

\begin{eqnarray}
\tilde{B} (w,z)-\tilde{B} (u,v) & = &\hspace{-0.2cm} \tilde{B} (g,g+\beta^{2}Ag+(\beta^{2}-\alpha^{2})Au)+\tilde{B} (g,v)\nonumber\\
& + &\hspace{-0.2cm} \tilde{B} (u,g+\beta^{2}Ag+(\beta^{2}-\alpha^{2})Au).\label{red}
\end{eqnarray}
Besides, using \eqref{zero}, we have that \eqref{red} reduces to 

$$(\tilde{B} (w,z)-\tilde{B} (u,v),g)
=(\tilde{B} (u,g),g)+\beta^{2}(\tilde{B} (u,Ag),g)+(\beta^{2}-\alpha^{2})(\tilde{B} (u,Au),g).
$$
Thus
\begin{align}
&\frac{1}{2}\frac{\mathrm{d} }{\mathrm{d} t}  (\left \| g \right \| ^{2}+\beta^{2}\left \| \nabla g \right \| ^{2})+\nu (\left \| \nabla g \right \| ^{2}+\beta^{2}\left \| A g \right \| ^{2})\nonumber\\
&\le  | \beta^{2}-\alpha^{2}  | \left | \left \langle \frac{\mathrm{d} }{\mathrm{d} t} Au,g  \right \rangle  \right | +\nu | \beta^{2}-\alpha^{2} |\,| (Au,Ag) | +| \tilde{B} (u,g),g|\nonumber \\
&+\beta^{2}\ | ( \tilde{B} (u,Ag),g ) | + | \beta^{2}-\alpha^{2}|\,| (\tilde{B}(u,Au),g)| -\eta (I_{h}g,g)-\eta \beta^{2}(I_{h}g,Ag).\label{stp1}
\end{align}

Using \eqref{Gagliardo-Nirenberg}, Young and Hölder's inequality, we estimate now the right-hand side terms above: 
\begin{align}
 | \beta^{2}-\alpha^{2}|\, \left | \left \langle \frac{\mathrm{d} }{\mathrm{d} t} Au,g  \right \rangle  \right | & \le | \beta^{2}-\alpha^{2}| \left \| u_{t} \right \| \left \| A_{g} \right \|\nonumber \\
&\le \frac{4 | \beta^{2}-\alpha^{2}  | ^{2} }{\nu \beta^{2}}\left \| u_{t} \right \|^{2}+\frac{\nu \beta^{2}}{16}\left \| Ag \right \| ^{2},  \label{es1}
\end{align}
\begin{align}
\nu  | \beta^{2}-\alpha^{2}  |\left | (Au,Ag) \right | \le \frac{\nu \beta^{2}}{16}\left \| Ag \right \| ^{2}+4\nu \frac{ | \beta^{2}-\alpha^{2}  |^{2} }{\beta^{2}} \left \| Au \right \|^{2}, \label{es2}
\end{align}
\begin{align}
 | ( \tilde{B} (u,g),g )| 
&\le c\left \| u \right \| _{L^{6}}\left \| \nabla g \right \| _{L^{2}}\left \| g \right \| _{L^{3}} \nonumber\\
&\le c\left \|\nabla  u \right \| _{L^{2}}\left \| \nabla g \right \| _{L^{2}} ^{\frac{3}{2}} \left \| g \right \| _{L^{2}} ^{\frac{1}{2}} \nonumber\\ 
&\le\frac{\nu }{16}\left \| \nabla g \right \| _{L^{2}} ^{2}+\frac{432c^{4}}{\nu ^{3}}   \left \| \nabla u \right \| ^{4} \left \| g \right \| ^{2}, \label{est3}
\end{align}
and
\begin{align}
\beta^{2} | ( \tilde{B} (u,Ag),g )| 
& \le c\beta^{2} \left \| u \right \| _{L^{6}}\left \|  Ag \right \| _{L^{2}}\left \| \nabla g \right \| _{L^{3}}\nonumber \\
& \le c\beta^{2} \left \| \nabla u \right \| _{L^{2}}\left \|  Ag \right \| ^{\frac{3}{2}} \left \| \nabla g \right \| ^\frac{1}{2}\nonumber \\
& \le\frac{\nu \beta^{2}}{16}\left \| Ag \right \| ^{2}+\frac{432c^{4} \beta^{2}}{\nu ^{3}}   \left \| \nabla u \right \| ^{4} \left \| \nabla g \right \| ^{2}.\label{est4}
\end{align}

Using integration by parts and (\ref{Poincare}), we obtain
\begin{align}
 | \beta^{2}-\alpha^{2}|\,| (\tilde{B}(u,Au),g)  |
& \le  | \beta^{2}-\alpha^{2}  | \left \| u \right \|_{L^{6}} \left \| Au \right \| \left \| \nabla g \right \|_{L^{3}} \nonumber\\
& \le  | \beta^{2}-\alpha^{2}  | c^{2}\left \| \nabla u \right \| \left \| Au \right \| \left \| \nabla g \right \| ^{\frac{1}{2}} \left \| Ag \right \|^{\frac{1}{2}} \nonumber\\
& \le \frac{4c^{4} }{\nu \beta^{2} \lambda_{1}^{\frac{1}{2}}}  | \beta^{2}-\alpha^{2}  | ^{2} \left \| \nabla u \right \|  ^{2} \left \| Au \right \| ^{2}+\frac{\nu \beta^{2} \lambda_{1}^{\frac{1}{2}}}{16 }\left \| \nabla g \right \| \left \| Ag \right \|\nonumber \\
& \le \frac{4c^{4}}{\nu \beta^{2} \lambda_{1}^{\frac{1}{2}}}  | \beta^{2}-\alpha^{2}  | ^{2} \left \| \nabla u \right \|  ^{2} \left \| Au \right \| ^{2} +\frac{\nu \beta^{2}}{16} \left \| Ag \right \| ^{2}.\label{est5}
\end{align}
 Finally, using condition \eqref{Ih}, it follows that 
\begin{align}
-\eta (I_{h}g,g)
& = -\eta (I_{h}g-g,g)-\eta \left \| g \right \| ^{2}\nonumber \\
& \le \eta \left \| I_{h}g-g \right \| ^{2}+\frac{\eta}{4} \left \| g \right \|^{2} -\eta \left \| g \right \| ^{2}\nonumber \\
& \le \eta(c_{1} h ^{2} \left \| \nabla g \right \| ^{2}+c_{2} h ^{4} \left \| Ag \right \| ^{2})-\frac{3}{4}\eta \left \| g \right \| ^{2},\label{est6}
\end{align}
and
\begin{align}
-\eta \beta^{2}(I_{h}g,Ag)
& = -\eta \beta^{2}(I_{h}g-g,Ag)-\eta \beta^{2}(g,Ag)\nonumber \\
& \le  \frac{\eta^{2}\beta^{2}}{\nu } (c_{1} h ^{2} \left \| \nabla g \right \| ^{2}+c_{2} h ^{4} \left \| Ag \right \| ^{2})+\frac{\nu }{4}\beta^{2}\left \| Ag \right \| ^{2}-\eta\beta^{2}\left \| \nabla g \right \| ^{2}.\label{est7}
\end{align}

Therefore, combining estimates \eqref{es1}-\eqref{est7} in \eqref{stp1}, we get
\begin{align}
&\frac{1}{2} \frac{\mathrm{d}}{\mathrm{d} t} (\| g \|^{2}+\beta^{2} \| \nabla g  \| ^{2}) \nonumber\\
&+\left(\frac{15}{16}\nu -\frac{432c^{8}\beta^{2}}{\nu ^{3}} \| \nabla u  \|^{4}-\eta c_{1}h^{2}-\frac{\eta^{2} c_{1} \beta^{2} h^{2}}{\nu}+\eta \beta^{2}\right)\| \nabla g \|^{2}\nonumber \\
& +\left(\frac{\nu}{2}\beta^{2}-\eta c_{2} h^{4}-\frac{\eta^{2}\beta^{2}c_{2}h^{4}}{\nu}\right)  \| Ag \ \| ^{2}+\left(\frac{3\eta}{4}-\frac{432c^{8}}{\nu^{3}} \| \nabla u\| ^{4}\right) \| g  \|^{2}\nonumber \\
& \le \frac{4 | \beta^{2}-\alpha^{2}|^{2}}{\nu \beta^{2}} \left \| u_{t} \right \|^{2}+\frac{4\nu}{\beta^{2}} |\beta^{2}-\alpha^{2}|^{2} \left \| Au \right \|^{2}+\frac{4c^{4}}{\nu \beta^{2} \lambda_{1}^{\frac{1}{2}}}  | \beta^{2}-\alpha^{2}|^{2}\left \| \nabla u \right \| ^{2}\left \| Au \right \|^{2}.\nonumber
\end{align}

By Lemma \ref{lem1}, we have that 
$$\|\nabla u\|^4\leq\ds\frac{M_1^2}{\alpha^4}, $$
Choosing $\eta \gg 1$ and $h\ll 1$ under the hypothesis listed in the statement of the Theorem, in addition to Poincaré's inequality (\ref{Poincare}), we obtain
\begin{align}
& \frac{1}{2} \frac{\mathrm{d}}{\mathrm{d} t} (\left \| g \right \|^{2}+\beta^{2}\left \| \nabla g \right \| ^{2})+\frac{\nu \lambda_{1}}{4}(\left \| g \right \|^{2}+\beta^{2}\left \| \nabla g \right \| ^{2}) \nonumber\\
& \le \frac{4 | \beta^{2}-\alpha^{2}  |^{2} }{\nu \beta^{2}} \left \| u_{t} \right \| ^{2}+\frac{4\nu}{\beta^{2}}  | \beta^{2}-\alpha^{2}  |^{2}\left \| Au \right \| ^{2}+\frac{4c^{4}}{\nu \beta^{2}\lambda_{1}^{\frac{1}{2}}}  | \beta^{2}-\alpha^{2}  |^{2}\left \| \nabla u \right \| ^{2}\left \| Au \right \| ^{2}. \label{usarGr}
\end{align}

Write $$\varphi(t)=4\frac{|\beta^2-\alpha^2|^2}{\beta^{2}} \left(\frac{1}{\nu }\|u_t\|^2+\nu\|Au\|^2+\frac{c^4}{\nu\lambda_{1}^{\frac{1}{2}} }\|\nabla u\|^2 \|Au\|^2\right).$$
From Lemmas \ref{lem1} and \ref{lem7}, we have that \\
\begin{align}
\ds\int_t^{t+T}&\varphi(s)\,ds  = \ds\frac{4|\beta^2-\alpha^2|^2}{\nu\beta^2}\int_t^{t+T}\|u_t(s)\|^2ds+\frac{4\nu|\beta^2-\alpha^2|^2}{\beta^2}\int_t^{t+T}\|Au(s)\|^2ds\,\nonumber\\
& +\ds\frac{4c^4|\beta^2-\alpha^2|^2}{\nu\beta^{2}\lambda_{1}^{\frac{1}{2}}}\ds\int_t^{t+T}\|\nabla u(s)\|^2\|Au(s)\|^2ds\nonumber\\
&\leq  \ds\frac{4|\beta^2-\alpha^2|^2}{\nu\beta^2}\Big[\frac{6}{\alpha^2}\Big(\nu^2+\frac{c^4M_1}{\lambda_1^{1/2}\alpha^2}\Big)\Big(\frac{M_1}{\nu}+\frac{T}{\nu^2\lambda_1}\sup\limits_{s\geq 0}\|f(s)\|^2 \Big)\nonumber\\
& +\frac{3c^4M_1^2T}{\alpha^6\lambda_1^{1/2}}+3T\sup\limits_{s\geq 0}\|f(s)\|^2\Big]+\frac{4\nu|\beta^2-\alpha^2|^2}{\beta^2\alpha^2}\Big(\frac{M_1}{\nu}+\frac{T}{\nu^2\lambda_1}\sup\limits_{s\geq 0}\|f(s)\|^2 \Big)\nonumber\\
& + \ds\frac{4c^4M_1|\beta^2-\alpha^2|^2}{\nu\beta^2\alpha^4 \lambda^{\frac{1}{2}}}\Big(\frac{M_1}{\nu}+\frac{T}{\nu^2\lambda_1}\sup\limits_{s\geq 0}\|f(s)\|^2 \Big).\nonumber
\end{align}
Choosing $T=\left(\nu \lambda_{1}\right)^{-1}>0$, we get
\begin{align}
\ds\int_t^{t+\frac{1}{\lambda_1\nu}}&\varphi(s)\,ds  \leq  \ds\frac{4|\beta^2-\alpha^2|^2}{\nu\beta^2}\Big[\frac{6}{\alpha^2}\Big(\nu^2+\frac{c^4M_1}{\lambda_1^{1/2}\alpha^2}\Big)\Big(\frac{M_1}{\nu}+\frac{1}{\nu^3\lambda_1^2}\sup\limits_{s\geq 0}\|f(s)\|^2 \Big)\nonumber\\
& +\frac{3c^4M_1^2}{\alpha^6\nu\lambda_1^{3/2}}+\frac{3}{\nu\lambda_1}\sup\limits_{s\geq 0}\|f(s)\|^2\Big]+\frac{4\nu|\beta^2-\alpha^2|^2}{\beta^2\alpha^2}\Big(\frac{M_1}{\nu}+\frac{1}{\nu^3\lambda_1^2}\sup\limits_{s\geq 0}\|f(s)\|^2 \Big)\nonumber\\
& + \ds\frac{4c^4M_1|\beta^2-\alpha^2|^2}{\nu\beta^2\alpha^4 \lambda^{\frac{1}{2}}}\Big(\frac{M_1}{\nu}+\frac{1}{\nu^3\lambda_1^2}\sup\limits_{s\geq 0}\|f(s)\|^2 \Big).\label{M222}
\end{align}

Denoting $\zeta(t)=\|g(t)\|^{2}+\beta^{2}\|\nabla g(t)\|^{2}, C=\ds\frac{\nu \lambda_{1}}{2}$, from \eqref{usarGr} and \eqref{M222}, we have
$$
\frac{d}{d t} \zeta(t)+C \zeta(t) \le 2\varphi(t),
$$
with $\ds\int_{t}^{t+\frac{1}{\lambda_{1} \nu}} 2\varphi(s) ds \le M$, where $M$ is defined  as
\begin{align}
M &=  \ds\frac{8|\beta^2-\alpha^2|^2}{\nu\beta^2}\Big[\frac{6}{\alpha^2}\Big(\nu^2+\frac{c^4M_1}{\lambda_1^{1/2}\alpha^2}\Big)\Big(\frac{M_1}{\nu}+\frac{1}{\nu^3\lambda_1^2}\sup\limits_{s\geq 0}\|f(s)\|^2 \Big)\nonumber\\
& +\frac{3c^4M_1^2}{\alpha^6\nu\lambda_1^{3/2}}+\frac{3}{\nu\lambda_1}\sup\limits_{s\geq 0}\|f(s)\|^2\Big]+\frac{8\nu|\beta^2-\alpha^2|^2}{\beta^2\alpha^2}\Big(\frac{M_1}{\nu}+\frac{1}{\nu^3\lambda_1^2}\sup\limits_{s\geq 0}\|f(s)\|^2 \Big)\nonumber\\
& + \ds\frac{8c^4M_1|\beta^2-\alpha^2|^2}{\nu\beta^2\alpha^4 \lambda_{1}^{\frac{1}{2}}}\Big(\frac{M_1}{\nu}+\frac{1}{\nu^3\lambda_1^2}\sup\limits_{s\geq 0}\|f(s)\|^2 \Big).\nonumber
\end{align}
 Therefore, since the conditions of Lemma \ref{Gronwall} are satisfied, we obtain (\ref{teo22}). 
\end{proof}

\section{Numerical Simulation}
In this section, we conduct numerical experiments to illustrate and verify the theoretical results on the convergence as stated in Theorem \ref{teoBardina} and Theorem \ref{teoNSalpha}. In order to complete this task, we numerically solve the simplified Bardina model and the Navier-Stokes-$\alpha$ model in a three-dimensional domain. Initially, we implemented a finite difference algorithm, which encountered stability issues. We then attempted to use finite element methods, but these also proved challenging, particularly in getting the nudging term to work effectively. After many tests, we adopted a newly developed Python package ``Dedalus", which supports symbolic entry for equations and conditions. Dedalus utilizes spectral methods for solving partial differential equations and is particularly convenient for problems with periodic boundary conditions. Using Dedalus, we performed the following two sets of numerical simulations for both the simplified Bardina model and the Navier-Stokes-$\alpha$ model: one is with initial conditions without a random component to assess the impact of $\eta$ and one is with initial conditions with a random component to assess the impact of $\eta$. For each scenario, we have provided two types of graphs: one presents the normalized difference between the solutions from the original system and the data assimilation system. In these graphs, a decreasing trend indicates convergence, while an increasing trend represents divergence. The other graph displays the velocity contours for both the original system and the data assimilation system at the initial and final time steps. In cases of convergence, even if the two systems start differently, their velocity contours become similar by the end. Conversely, in divergent cases, the velocity contours remain distinct. The following provide the details.

\subsection{Testing the impact of $\eta$-without random initial conditions}
\subsubsection{Numerical Simulation for the simplified Bardina model}

Choosing a domain $\Omega=[0, 1]^3$, we have the initial conditions for the original system (\ref{Bardina1}) to be $u=(u_0,\ v_0,\ w_0):$
$$u_0=0.1 * \sin x,$$
$$v_0=-0.05 * \sin y,$$
$$w_0=-0.05 * \sin z.$$

The initial conditions for the assimilated model (\ref{Bardina1assimilated}) is taken to be $w=(\hat{u}_0,\ \hat{v}_0,\ \hat{w}_0):$
$$\hat{u}_0=-0.025 * \sin (4x),$$
$$\hat{v}_0=0.025 * \sin (4y),$$
$$\hat{w}_0=0.025 * \sin (4z).$$
In practice, this initial condition is arbitrary and can be set to anything within the domain.

We carefully select the parameters so they satisfy hypotheses 1 and 2 given in Theorem \ref{teoBardina}, i.e. 
\begin{enumerate}
    \item $\eta> \ds\frac{27c^4}{16}\left(\frac{3}{\nu}\right)^3\ds\frac{M_1^2}{\alpha^4}\ \Rightarrow \ \eta>C_1:=\frac{27c^4}{16}\left(\frac{3}{\nu}\right)^3\ds\frac{M_1^2}{\alpha^4};$\label{cond1a}
        \item $2\eta c_1 h^2<\nu\ \Rightarrow \nu>C_2:= 2\eta c_1 h^2;$\\
        \item $2\eta c_2h^4<\nu\beta^2\ \Rightarrow \beta>C_3:=\sqrt{\frac{2\eta c_2h^4}{\nu}}$.
         \label{cond2a}
           \end{enumerate}
Here, the constants are $c={\frac{4}{3\sqrt{3}}}^{3/4}$ \cite{Galdi}, $c_1=\sqrt{32}$ and $c_2=2$ \cite{albanez2016continuous}. We first fix $\nu=0.45$ and $\alpha=0.25$. $M_1$ depends on the initial conditions and force. Here, we take the force to be 0 and $M_1=0.003399$, so we have $C_1=0.91399$. We then compare two scenarios: one with $\eta=4>C_1$
  and the other with $\eta=0.0001<C_1$. Once $\eta$ is chosen, we can choose the $h=\frac{1}{39}$ value to make $\nu=0.45>C_2$. Lastly, we choose $\beta=0.3$ so $\beta>C_3$. 
  The graphical results on the difference are presented in Figures 1 (high $\eta$) and 4 (low $\eta$). In all these difference figures, the x-axis represents time, while the y-axis shows the logarithm of the normalized difference between the solution from the original system and the data assimilation system. The graphical results on the velocity contour are presented in Figures 2, 3 (high $\eta$)
  and 5, 6 (low $\eta$). In all these contour graphs, we compare the original and the assimilated systems at the beginning and at the end time step.


\begin{figure}
\centering
    \includegraphics[totalheight=7cm]{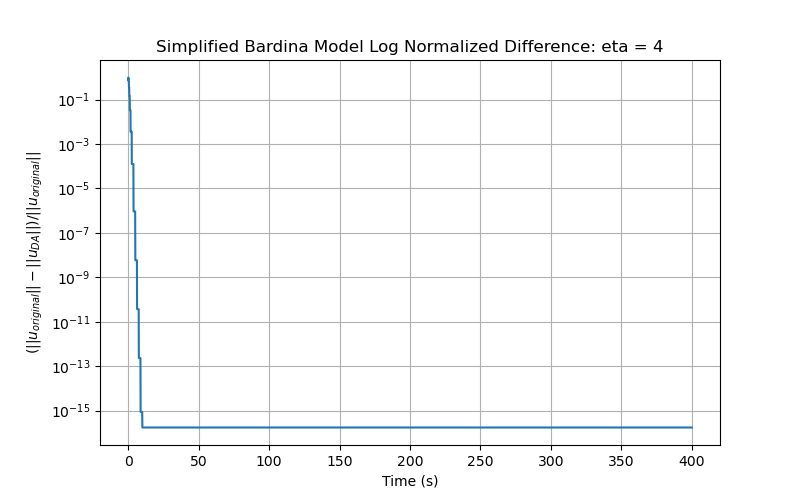}
    \caption{Error plot of Bardina model with high $\eta$ value-without random initial conditions case.}
    \label{fig1}
\end{figure}
\begin{figure}
\centering
    \includegraphics[totalheight=7cm]{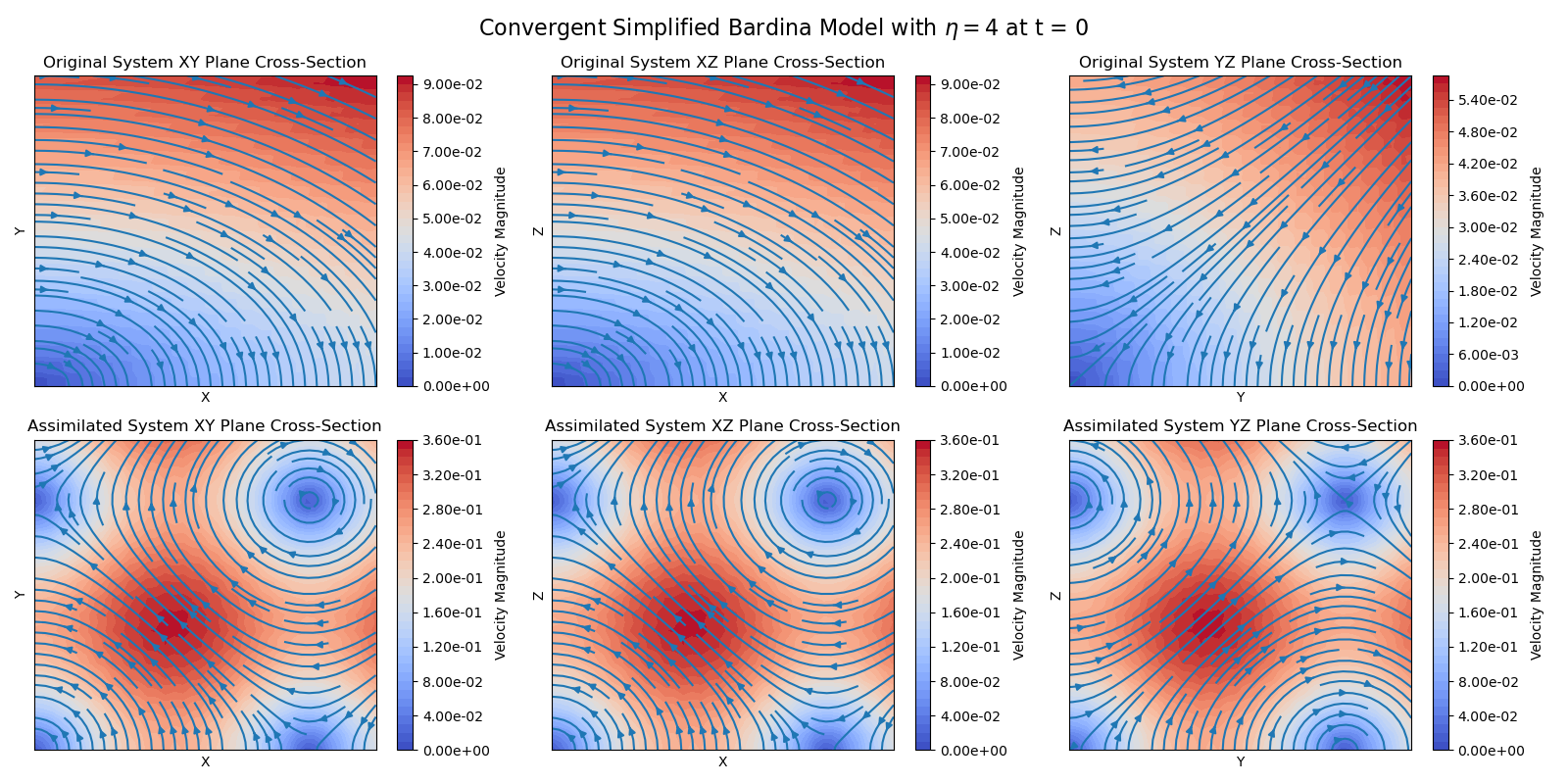}
    \caption{Velocity contour of Bardina model with high $\eta$ value-without random initial conditions case at $t=0$.}
    \label{fig2}
\end{figure}
\begin{figure}
\centering
    \includegraphics[totalheight=7cm]{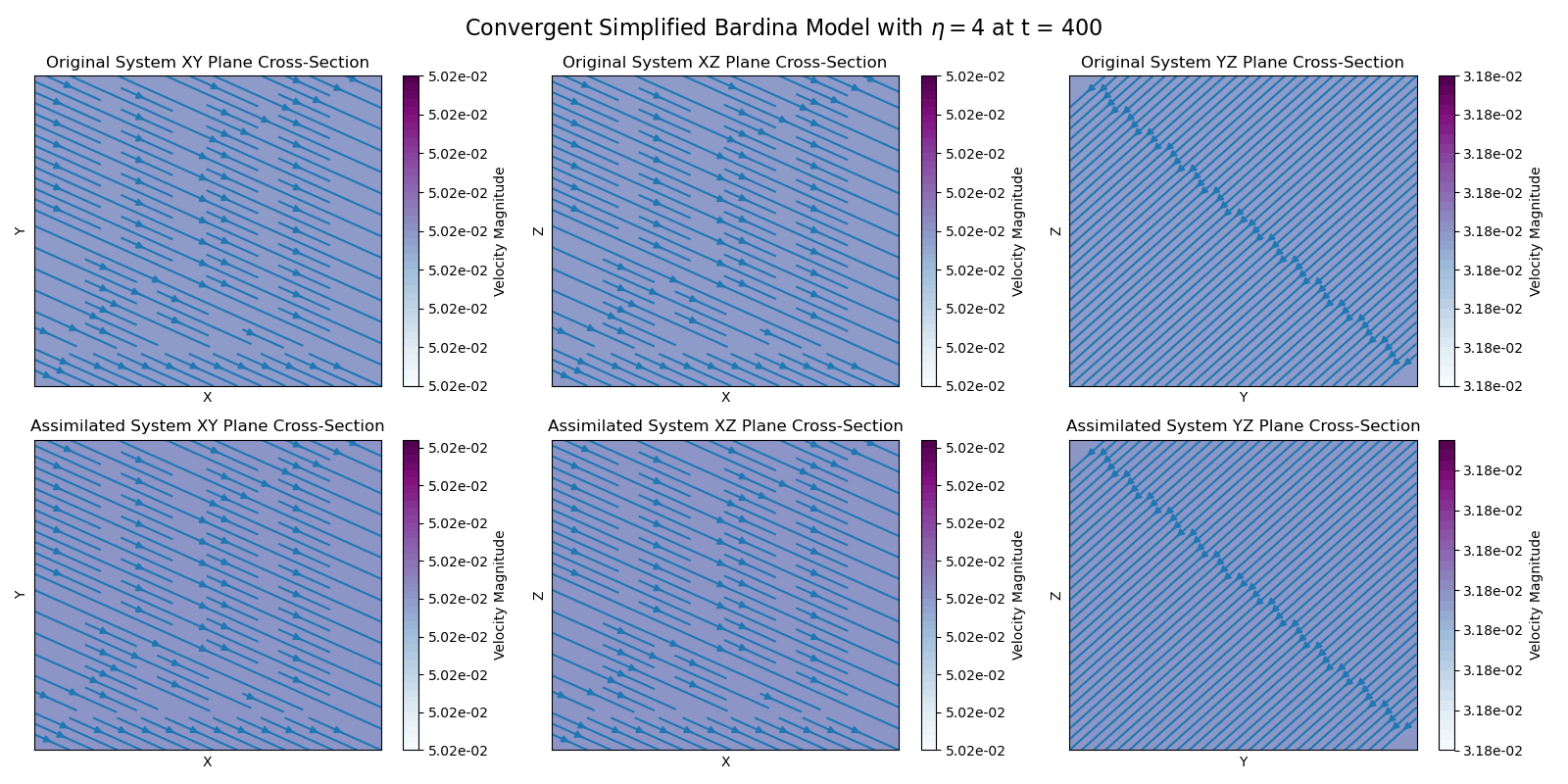}
    \caption{Velocity contour of Bardina model with high $\eta$ value-without random initial conditions case at $t=400$.}
    \label{fig3}
\end{figure}
\begin{figure}
\centering
    \includegraphics[totalheight=7cm]{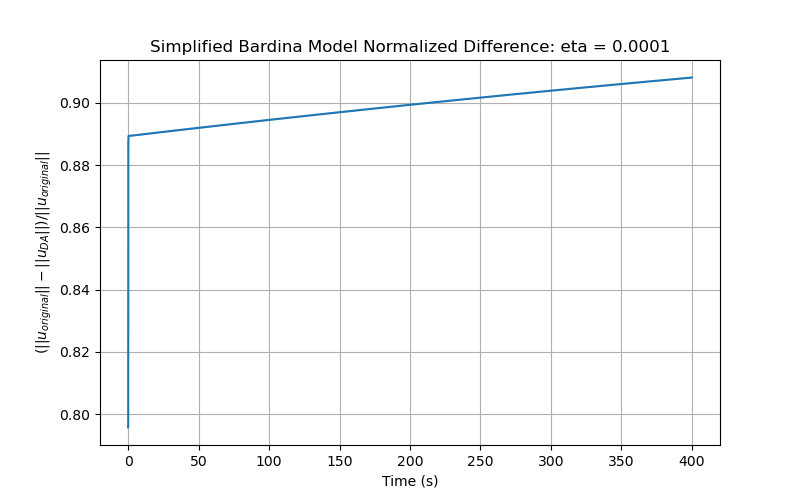}
    \caption{Error plot of Bardina model with high $\eta$ value-without random initial conditions case.}
    \label{fig4}
\end{figure}
\begin{figure}
\centering
    \includegraphics[totalheight=7cm]{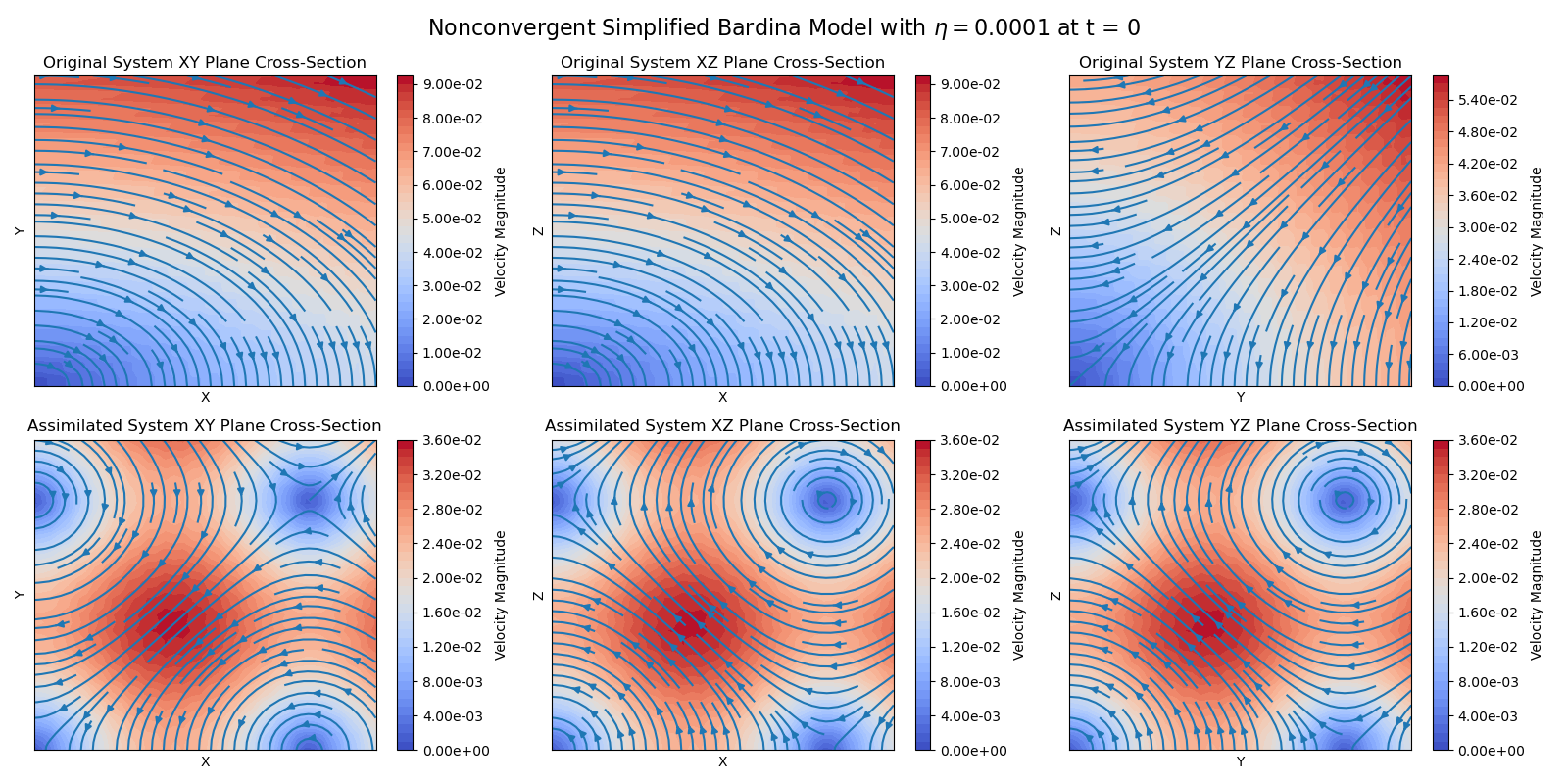}
    \caption{Velocity contour of Bardina model with high $\eta$ value-without random initial conditions case at $t=0$.}
    \label{fig5}
\end{figure}
\begin{figure}
\centering
    \includegraphics[totalheight=7cm]{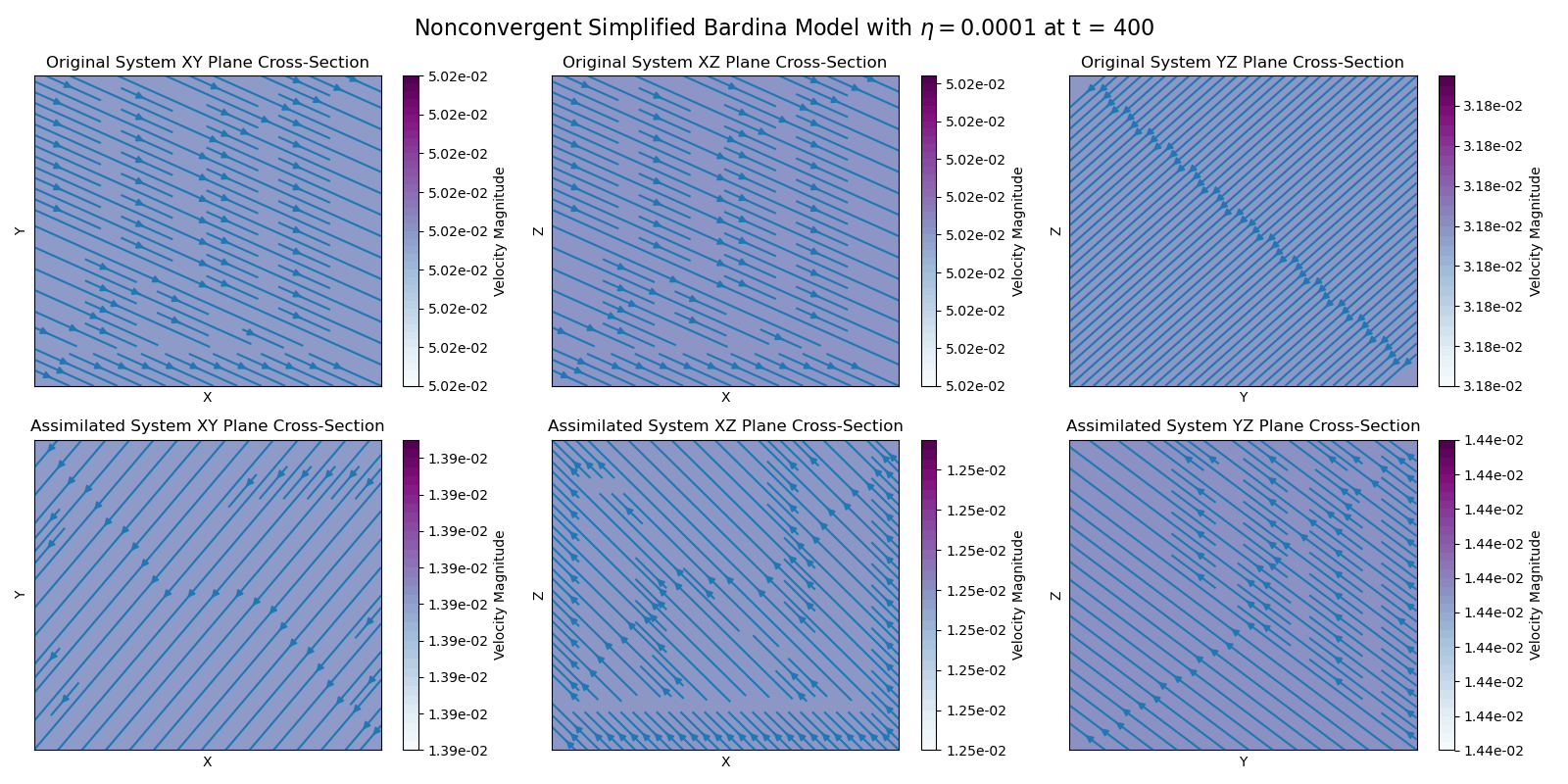}
    \caption{Velocity contour of Bardina model with high $\eta$ value-without random initial conditions case at $t=400$.}
    \label{fig6}
\end{figure}
\subsubsection{Numerical Simulation for the Navier-Stokes-$\alpha$ model}
When comparing the simplified Bardina (\ref{Bardina1}) and the Navier-Stokes-$\alpha$ (\ref{NSalpha1}) model, the key difference lies in their nonlinear terms. In the application of the data assimilation algorithm, the nudging is applied to the unfiltered velocity for the Navier-Stokes-$\alpha$ model, whereas in the Bardina model,  it is applied to the filtered velocity.

For the numerical computation, we have the domain to be $[0, 1]^3$ and we use the same initial conditions as the simplified Bardina model.

We also need to choose $\eta$ large enough and $h$ small enough so they satisfy hypotheses 1, 2 and 3 given in Theorem \ref{teoNSalpha}, i.e. 
\begin{enumerate}
\item $\eta\ge \frac{576c^{8}M_{1}^{2}}{\alpha^{4} \nu^{3}}\ \Rightarrow \eta \geq C_1:=\frac{576c^{8}M_{1}^{2}}{\alpha^{4} \nu^{3}}$, \label{hyp1}
\item $\eta c_{1}h^{2}+\ds\frac{\eta^{2}c_{1}\beta^{2}h^{2}}{\nu}+\frac{432c^{8}\beta^{2}M^{2}_{1}}{\nu^{3}\alpha^{4}}-\eta\beta^{2}<\frac{11\nu}{16}$, \\
\item $\eta c_{2}h^{4}+\ds\frac{\eta^{2}\beta^{2}c_{2}h^{4}}{\nu}< \frac{\nu}{4}\beta^{2}.$ 
 \end{enumerate}
Here, the constants are $c={\frac{4}{3\sqrt{3}}}^{3/4}$ \cite{Galdi}, $c_1=\sqrt{32}$ and $c_2=2$ \cite{albanez2016continuous}. We first fix the $\nu=0.45$ and $\alpha=0.25$. $M_1$ depends on the initial conditions and force. Here we take the force to be 0 and $M_1=0.00264$. We have $C_1=2.3546$. We compare two cases: $\eta=4>C_1$ (see Figures 7-9) and $\eta=0.0001<C_1$ (see Figures 10-12). Once $\eta$ is chosen, we can choose the $h=\frac{1}{39}$ and $\beta=0.3$ so the second and third conditions above are satisfied. 
\begin{figure}
\centering
    \includegraphics[totalheight=7cm]{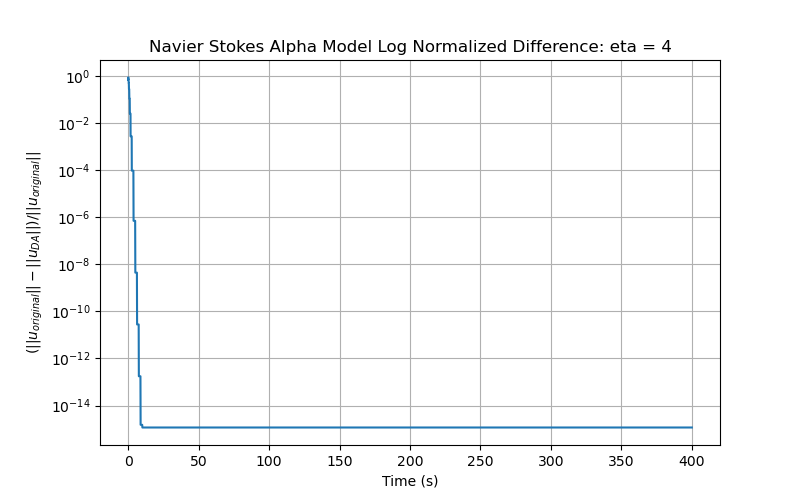}
    \caption{Error plot of NS-$\alpha$ model with high $\eta$ value-without random initial conditions case.}
    \label{fig7}
\end{figure}
\begin{figure}
\centering
    \includegraphics[totalheight=7cm]{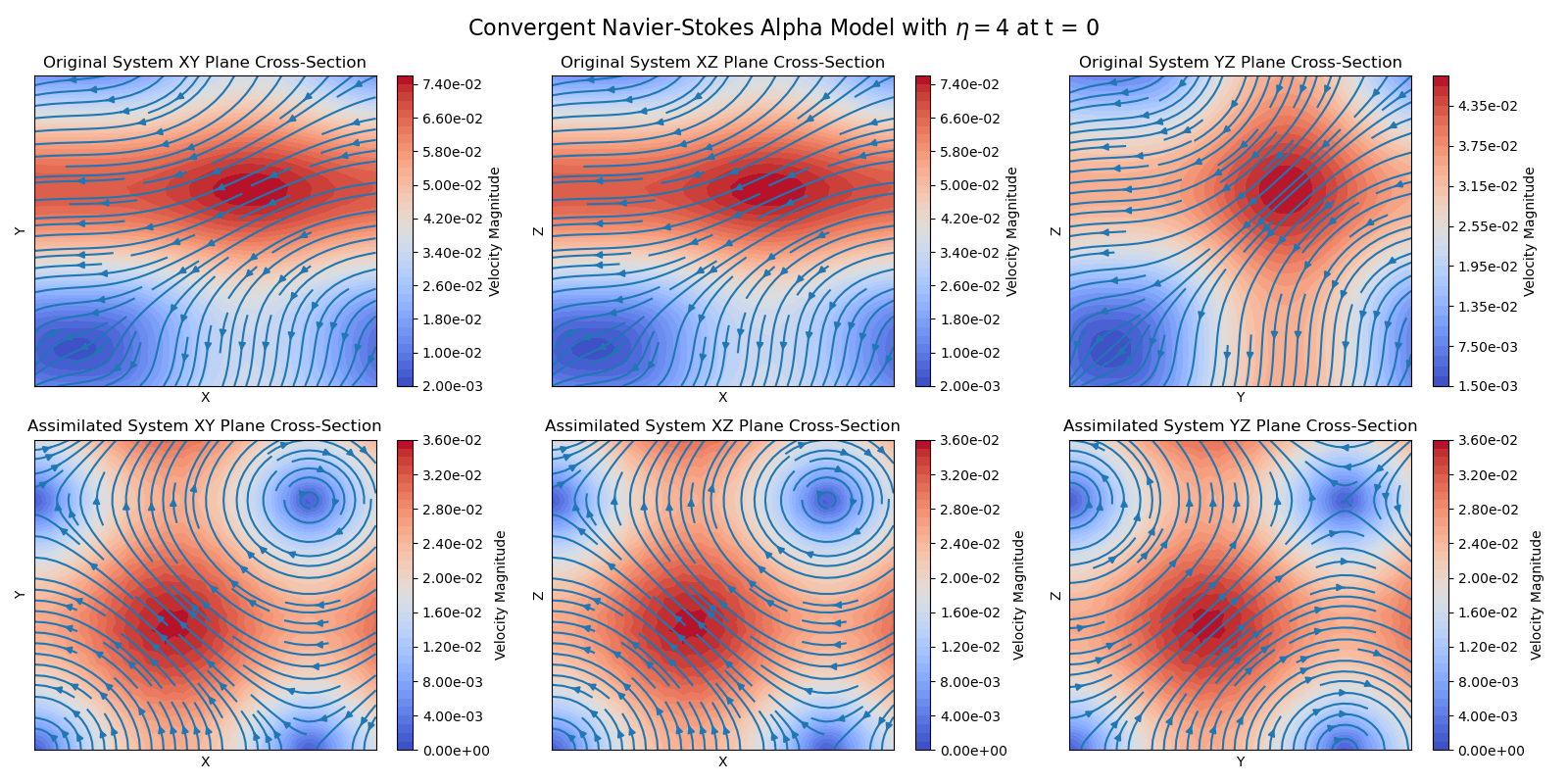}
    \caption{Velocity contour of NS-$\alpha$ with high $\eta$ value-without random initial conditions case at $t=0$.}
    \label{fig8}
\end{figure}
\begin{figure}
\centering
    \includegraphics[totalheight=7cm]{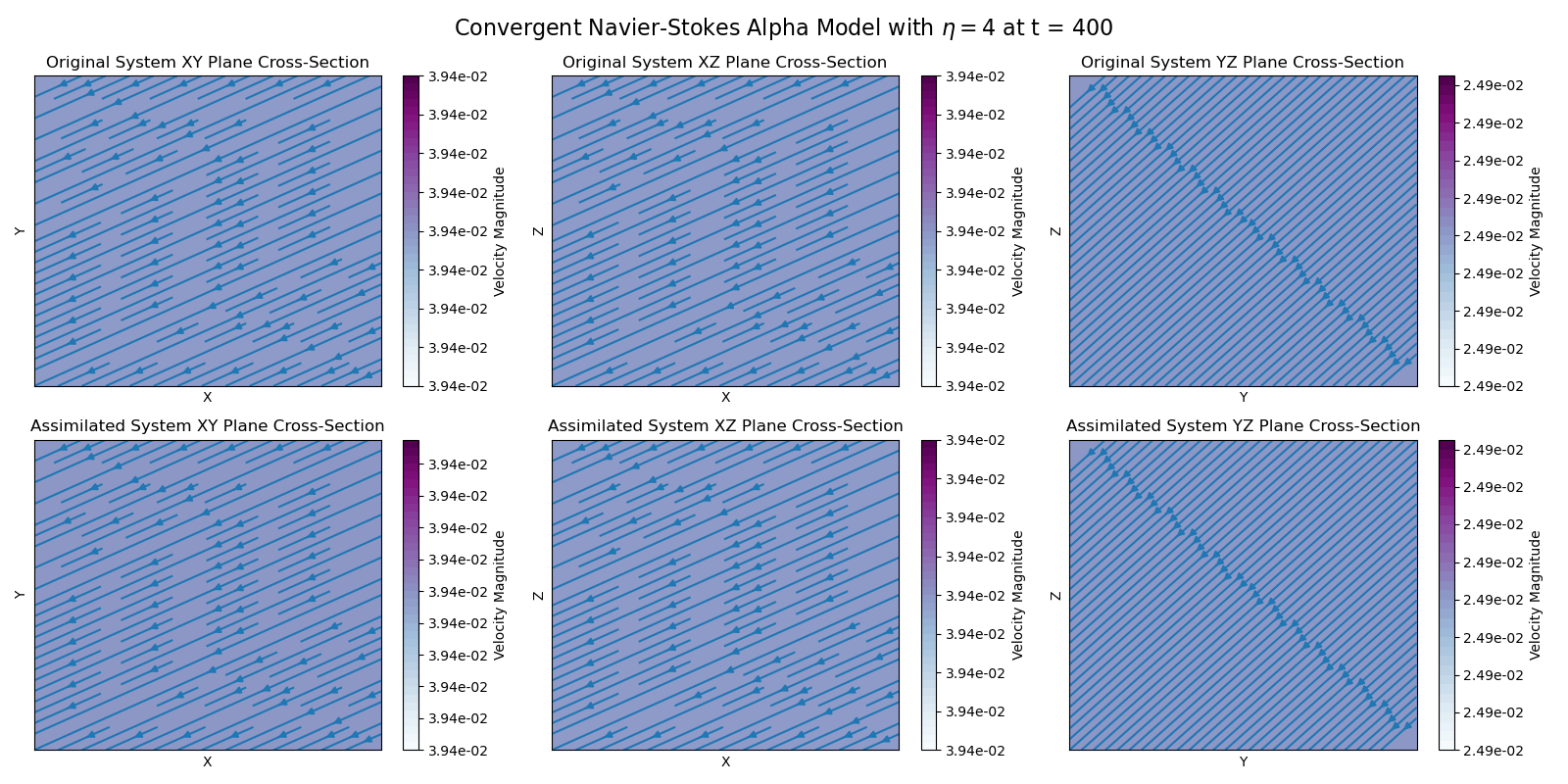}
    \caption{Velocity contour of NS-$\alpha$ model with high $\eta$ value-without random initial conditions case at $t=400$.}
    \label{fig9}
\end{figure}
\begin{figure}
\centering
    \includegraphics[totalheight=7cm]{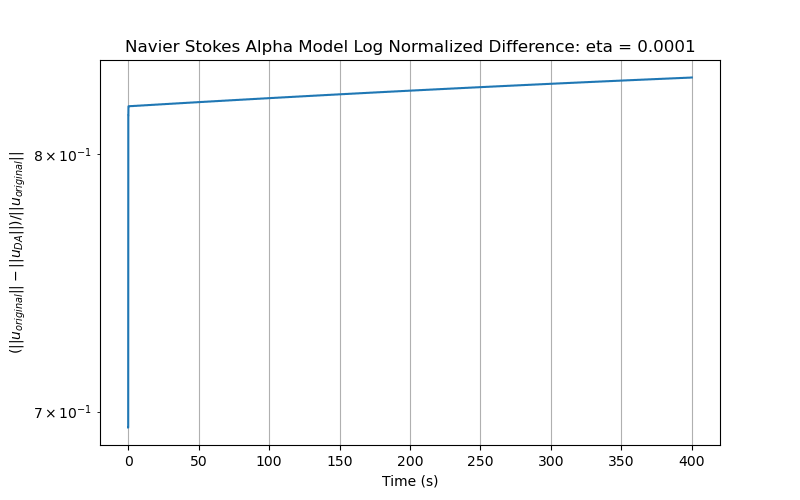}
    \caption{Error plot of NS-$\alpha$ model with high $\eta$ value-without random initial conditions case.}
    \label{fig10}
\end{figure}
\begin{figure}
\centering
    \includegraphics[totalheight=7cm]{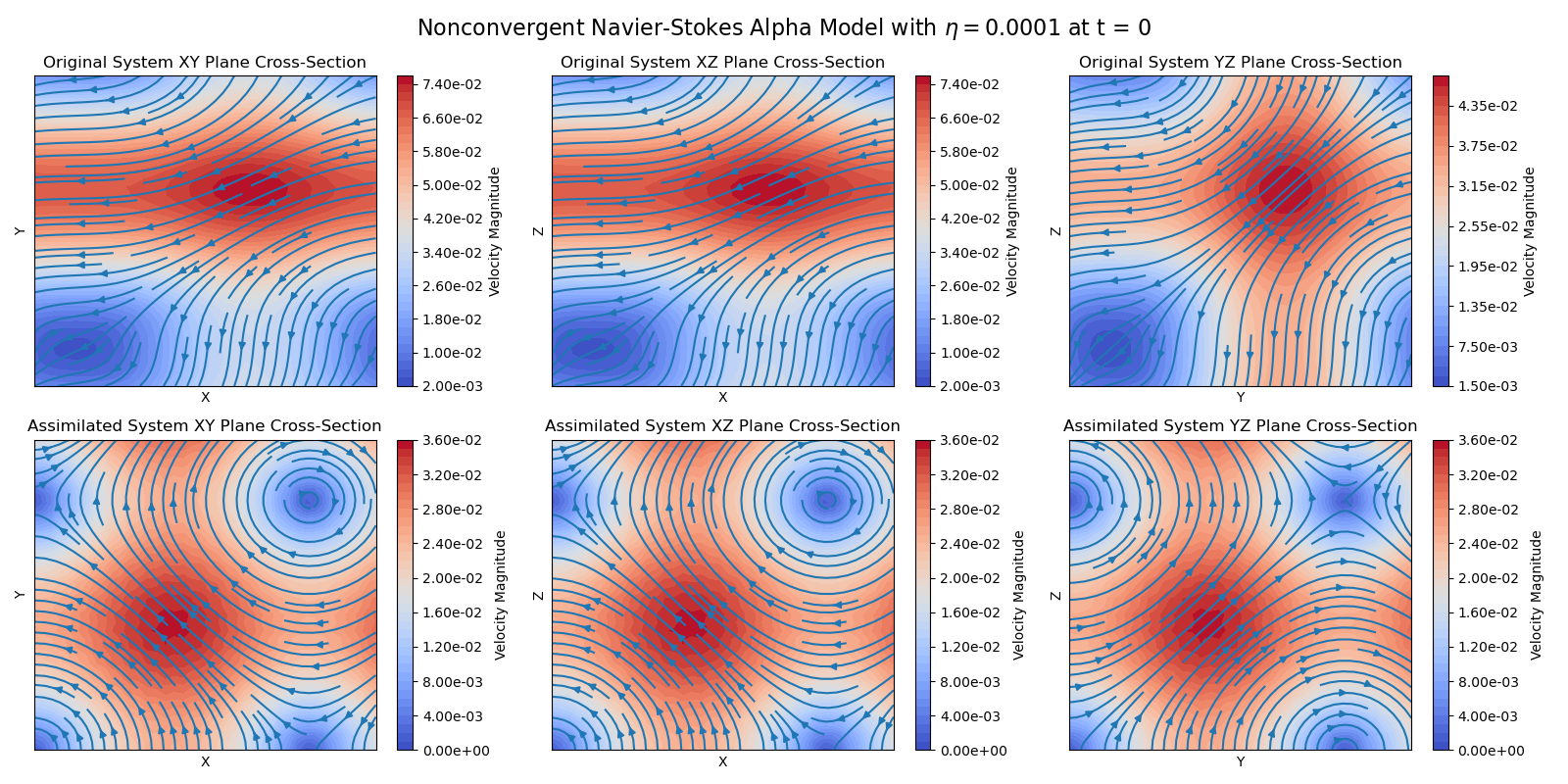}
    \caption{Velocity contour of NS-$\alpha$ model with high $\eta$ value-without random initial conditions case at $t=0$.}
    \label{fig11}
\end{figure}
\begin{figure}
\centering
    \includegraphics[totalheight=7cm]{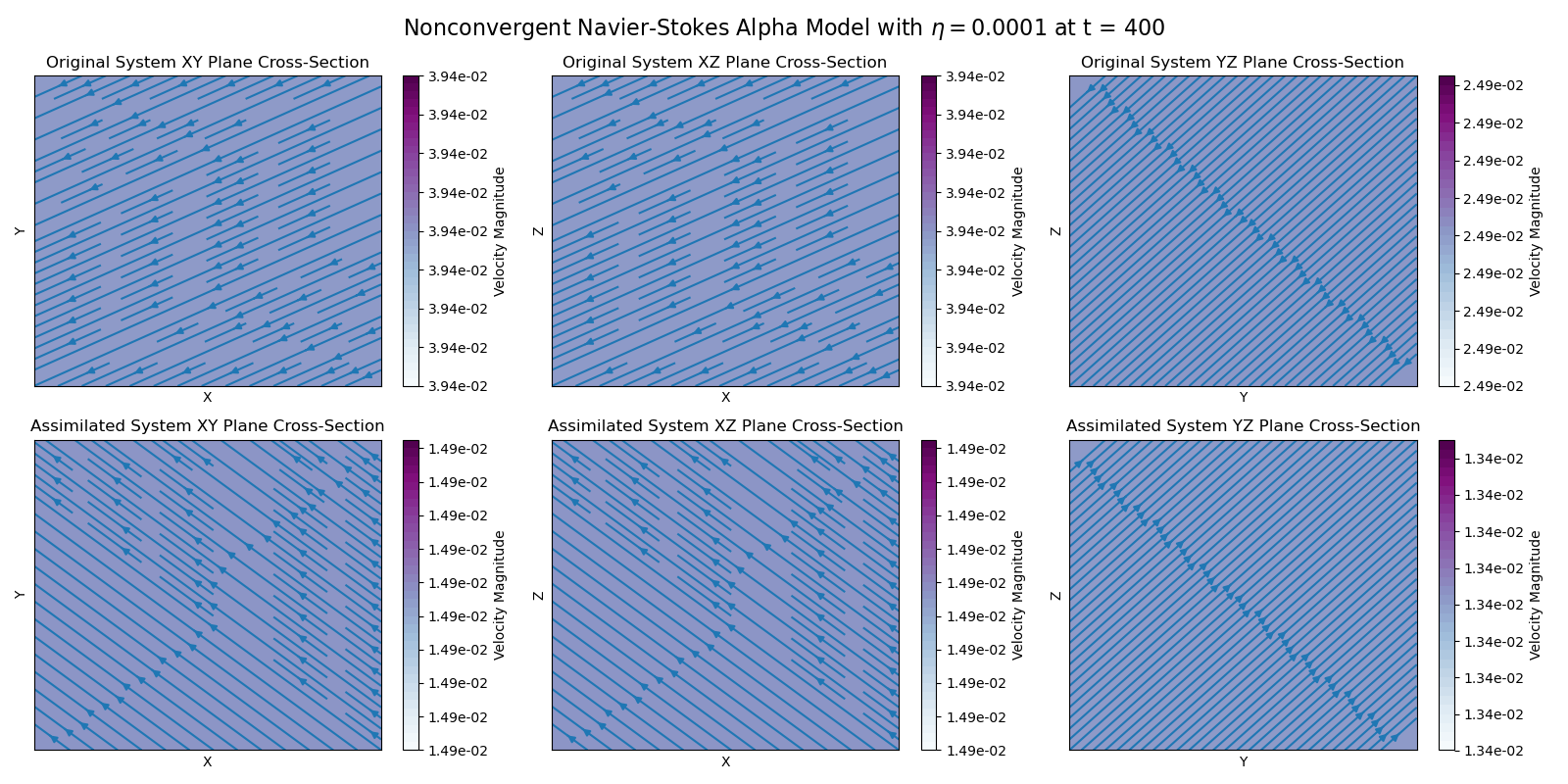}
    \caption{Velocity contour of NS-$\alpha$ model with high $\eta$ value-without random initial conditions case at $t=400$.}
    \label{fig12}
\end{figure}
\subsection{Testing the impact of $\eta$-with random initial conditions}
\subsubsection{Numerical Simulation for the simplified Bardina model}
The domain is $\Omega=[0, 1]^3$ and the initial conditions for the original system (\ref{Bardina1}) has a random component
 which is $u=(u_0,\ v_0,\ w_0):$
$$u_0=0.1*\sin x+0.01*X -0.005,$$
$$v_0=-0.05*\sin y+0.01*X -0.005,$$
$$w_0=-0.05*\sin z+0.01*X -0.005,$$
where $X$ is a random variable drawn from a uniform distribution.

The initial conditions for the assimilated model (\ref{Bardina1assimilated}) is taken to be $w=(\hat{u}_0,\ \hat{v}_0,\ \hat{w}_0):$
$$\hat{u}_0=-0.025 * \sin (4x),$$
$$\hat{v}_0=0.025 * \sin (4y),$$
$$\hat{w}_0=0.025 * \sin (4z).$$
Here, we have $\nu=0.45$ and $\alpha=0.25$. $M_1=0.00336$, $h=\frac{1}{39}$, and $\beta=0.3$. We compare the results when $\eta=4>C_1\approx 0.6623$ (results in 13-15) and $\eta=0.0001<C_1\approx 0.6623$ (results in 16-18). 
Note that, due to the random component in our initial conditions, each run yields a different $C_1$ value; however, these values are very close to one another.
\begin{figure}
\centering
    \includegraphics[totalheight=7cm]{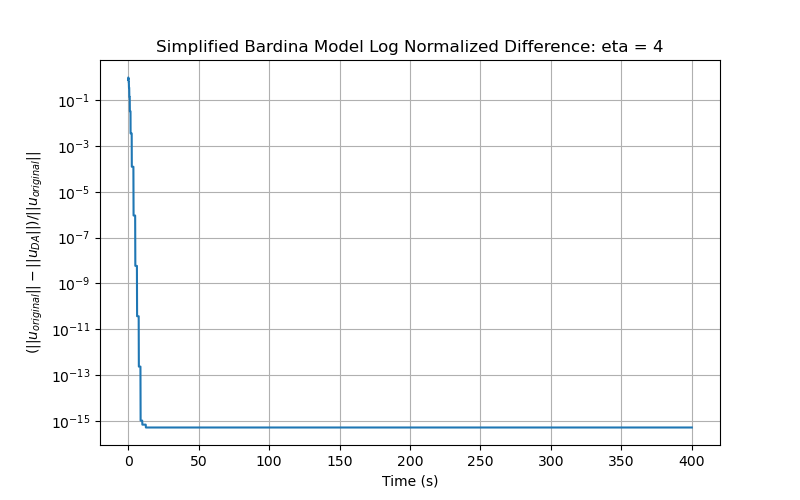}
    \caption{Error plot of Bardina model with high $\eta$ value-with random initial conditions case.}
    \label{fig13}
\end{figure}
\begin{figure}
\centering
    \includegraphics[totalheight=7cm]{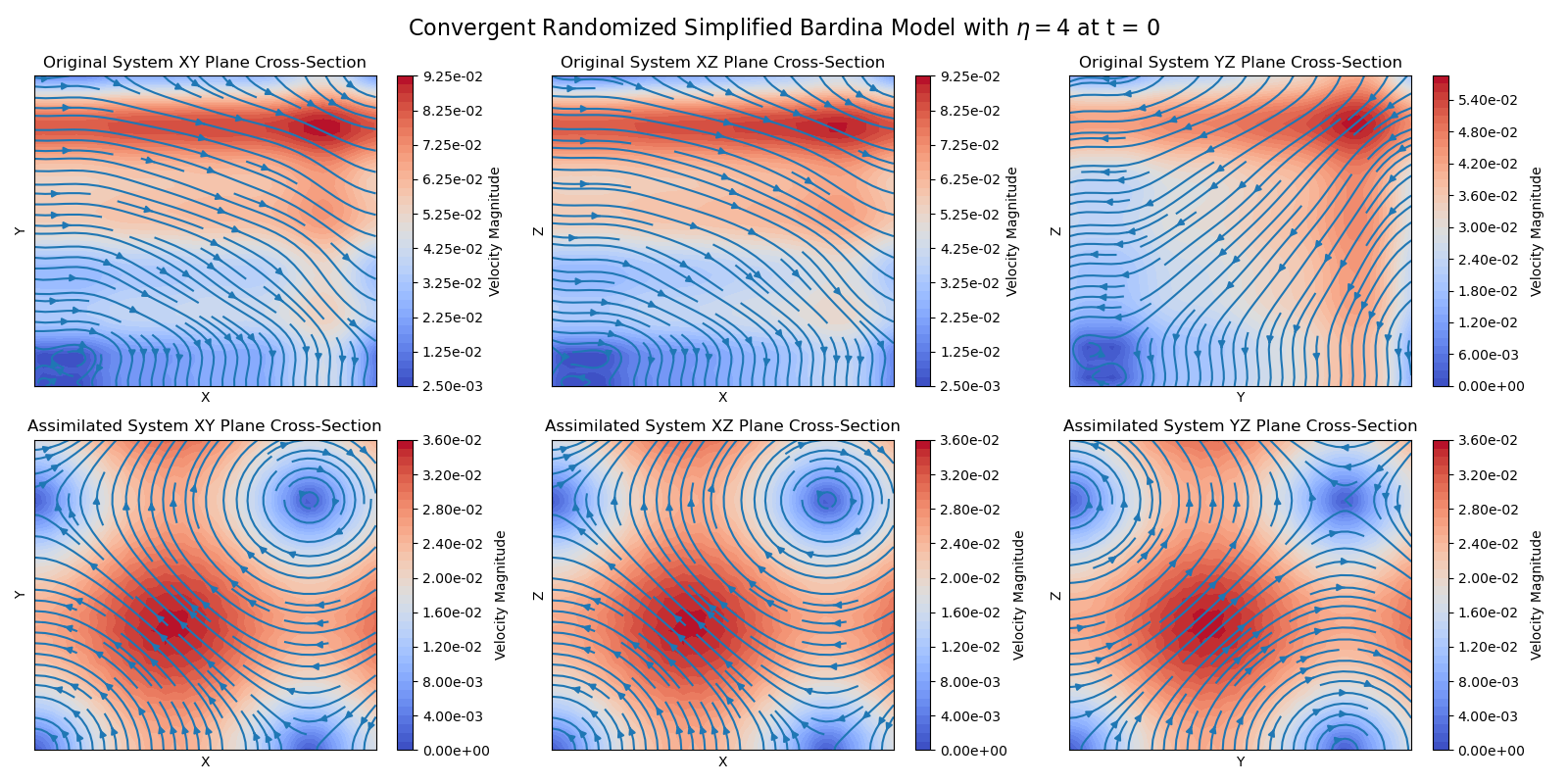}
    \caption{Velocity contour of Bardina model with high $\eta$ value-with random initial conditions case at $t=0$.}
    \label{fig14}
\end{figure}
\begin{figure}
\centering
    \includegraphics[totalheight=7cm]{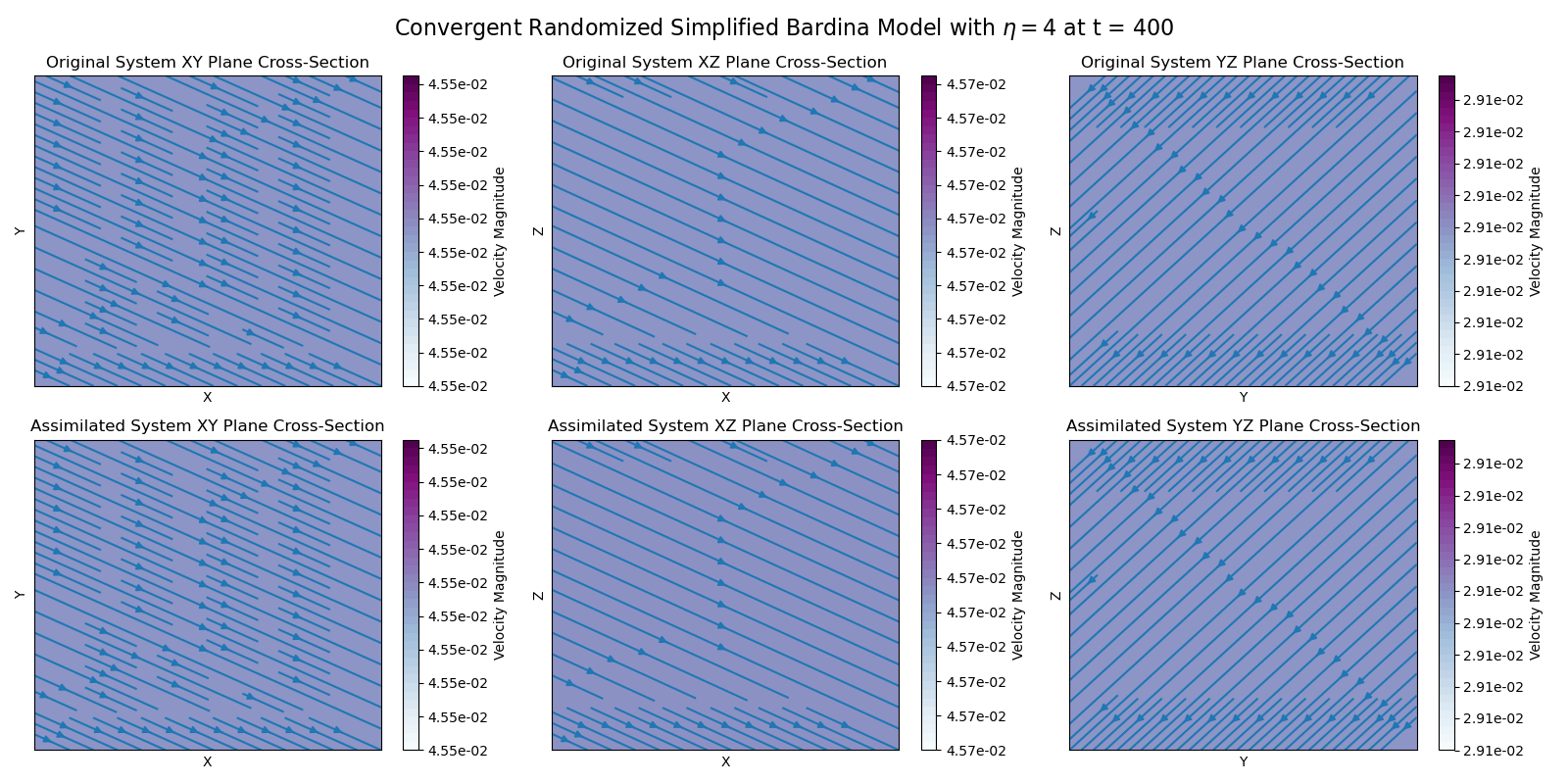}
    \caption{Velocity contour of Bardina model with high $\eta$ value-with random initial conditions case at $t=400$.}
    \label{fig15}
\end{figure}
\begin{figure}
\centering
    \includegraphics[totalheight=7cm]{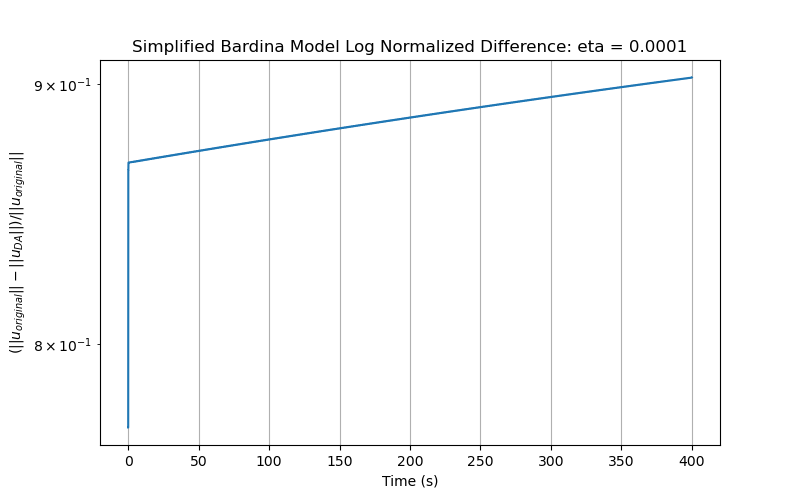}
    \caption{Error plot of Bardina model with high $\eta$ value-with random initial conditions case.}
    \label{fig16}
\end{figure}
\begin{figure}
\centering
    \includegraphics[totalheight=7cm]{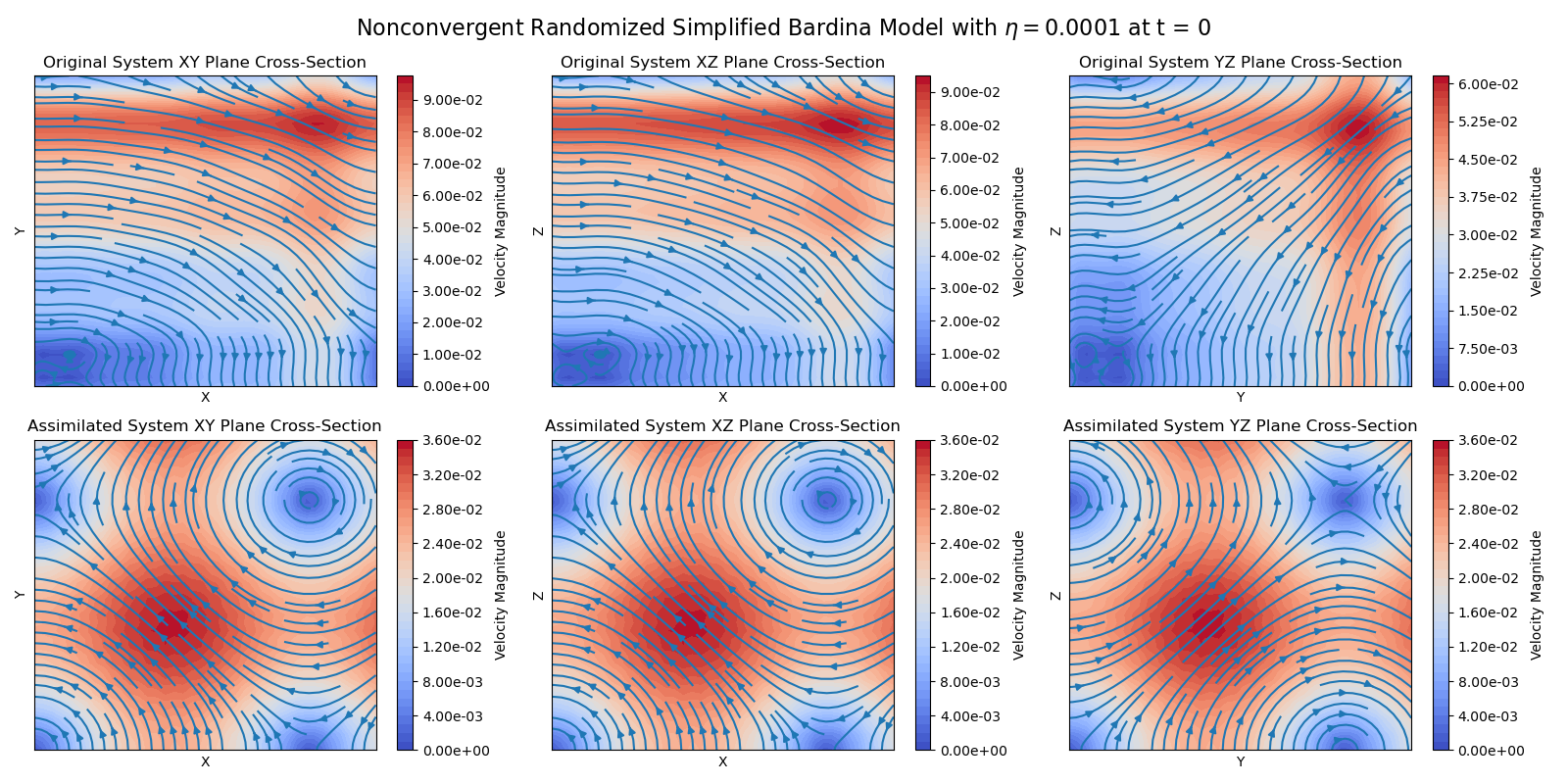}
    \caption{Velocity contour of Bardina model with high $\eta$ value-with random initial conditions case at $t=0$.}
    \label{fig17}
\end{figure}
\begin{figure}
\centering
    \includegraphics[totalheight=7cm]{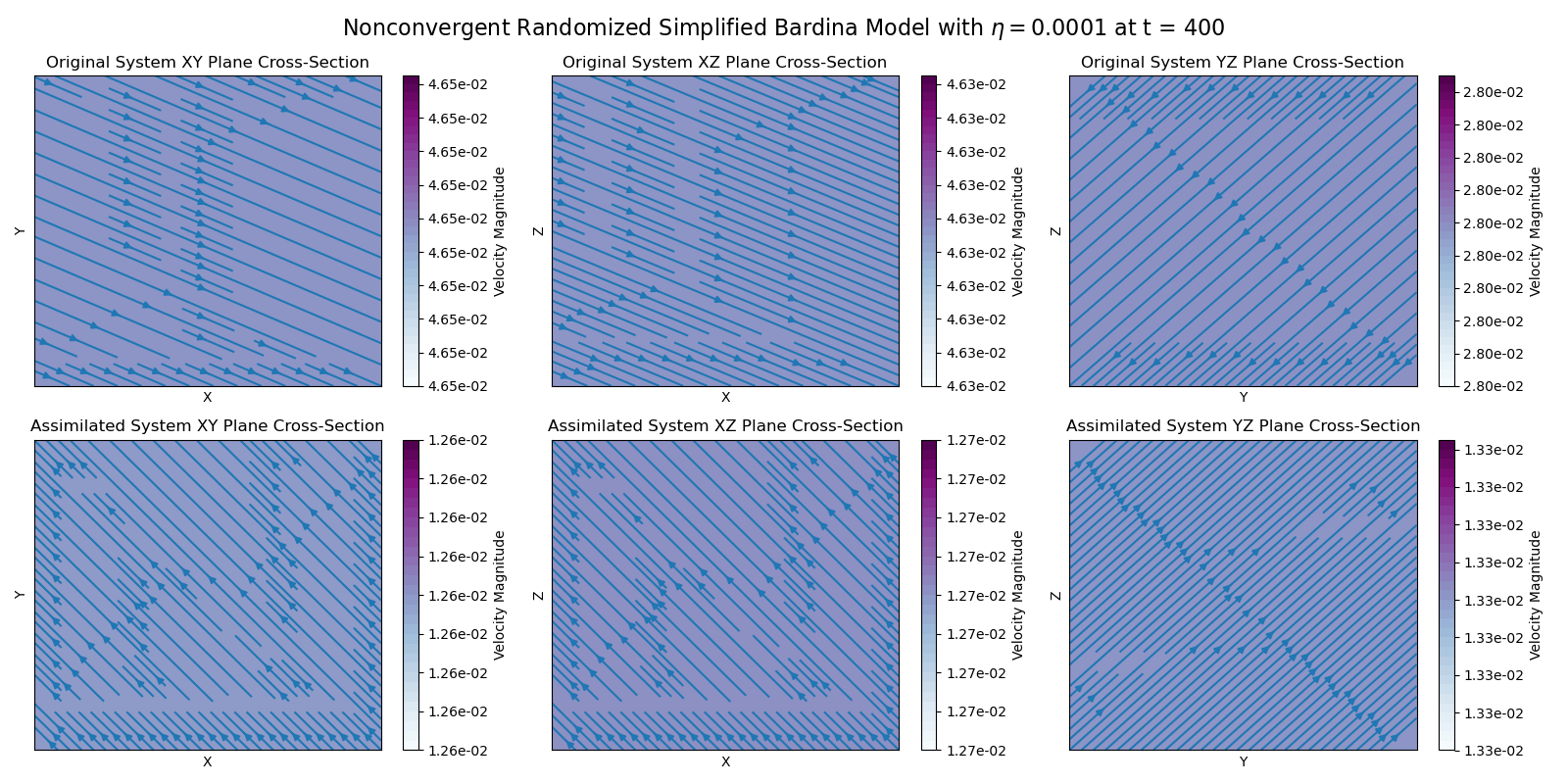}
    \caption{Velocity contour of Bardina model with high $\eta$ value-with random initial conditions case at $t=400$.}
    \label{fig18}
\end{figure}
\subsubsection{Numerical Simulation for the Navier-Stokes-$\alpha$ model}
The domain is $\Omega=[0, 1]^3$ and the initial conditions for the original system (\ref{NSalpha1}) has a random component
 which is $u=(u_0,\ v_0,\ w_0):$
$$u_0=-0.1*\sin x+0.01*X -0.005,$$
$$v_0=-0.05 * \sin y+0.01*X-0.005,$$
$$w_0=-0.05 * \sin z+0.01*X-0.005,$$
where $X$ is a random variable drawn from a uniform distribution.

The initial conditions for the assimilated model (\ref{NSalpha1assimilated}) is taken to be $w=(\hat{u}_0,\ \hat{v}_0,\ \hat{w}_0):$
$$\hat{u}_0=-0.025 * \sin (4x),$$
$$\hat{v}_0=0.025 * \sin (4y),$$
$$\hat{w}_0=0.025 * \sin (4z).$$
Here, we have $\nu=0.45$ and $\alpha=0.25$. $M_1=0.003384$, $h=\frac{1}{39}$, and $\beta=0.3$. We compare the results when $\eta=4>C_1\approx 3.85723$ (results in Figure 19-21) and $\eta=0.0001<C_1\approx 3.85723$ (results in Figure 22-24). 
\begin{figure}
\centering
    \includegraphics[totalheight=7cm]{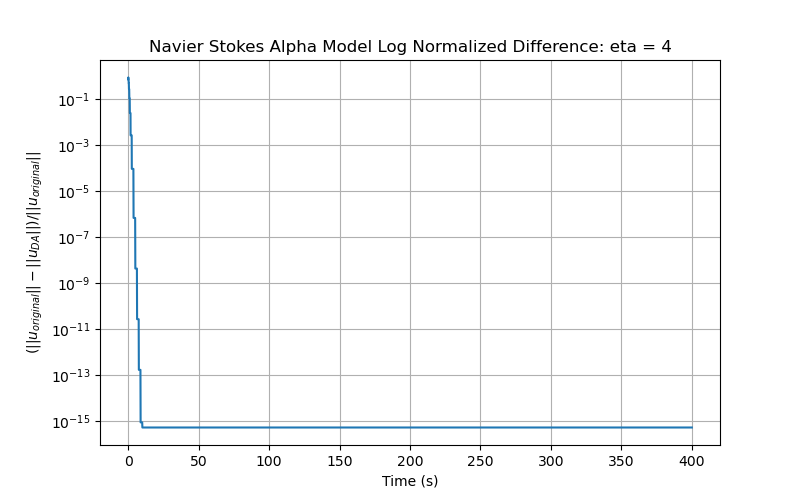}
    \caption{Error plot of NS-$\alpha$ model with high $\eta$ value-with random initial conditions case.}
    \label{fig19}
\end{figure}
\begin{figure}
\centering
    \includegraphics[totalheight=7cm]{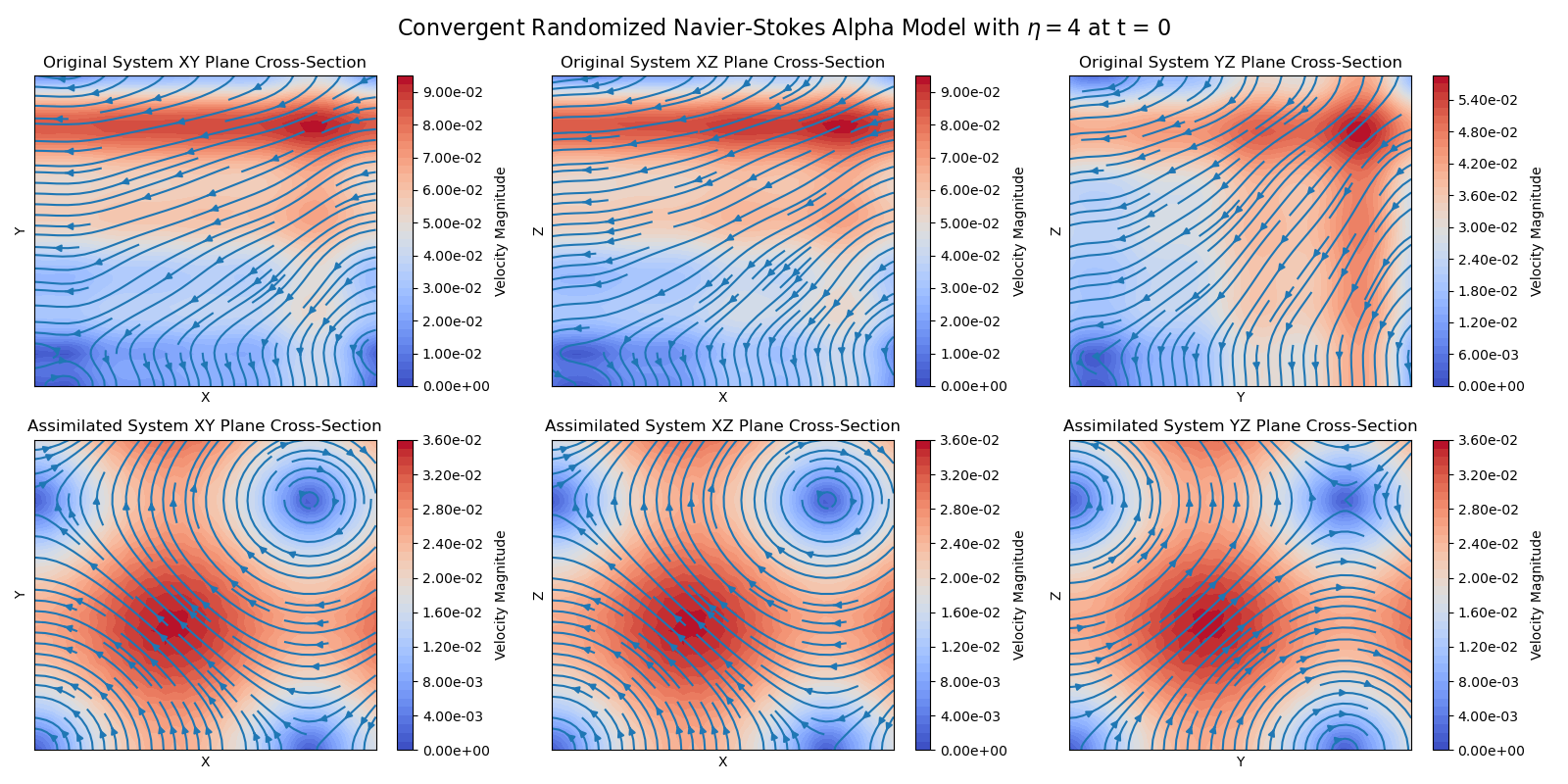}
    \caption{Velocity contour of NS-$\alpha$ with high $\eta$ value-with random initial conditions case at $t=0$.}
    \label{fig20}
\end{figure}
\begin{figure}
\centering
    \includegraphics[totalheight=7cm]{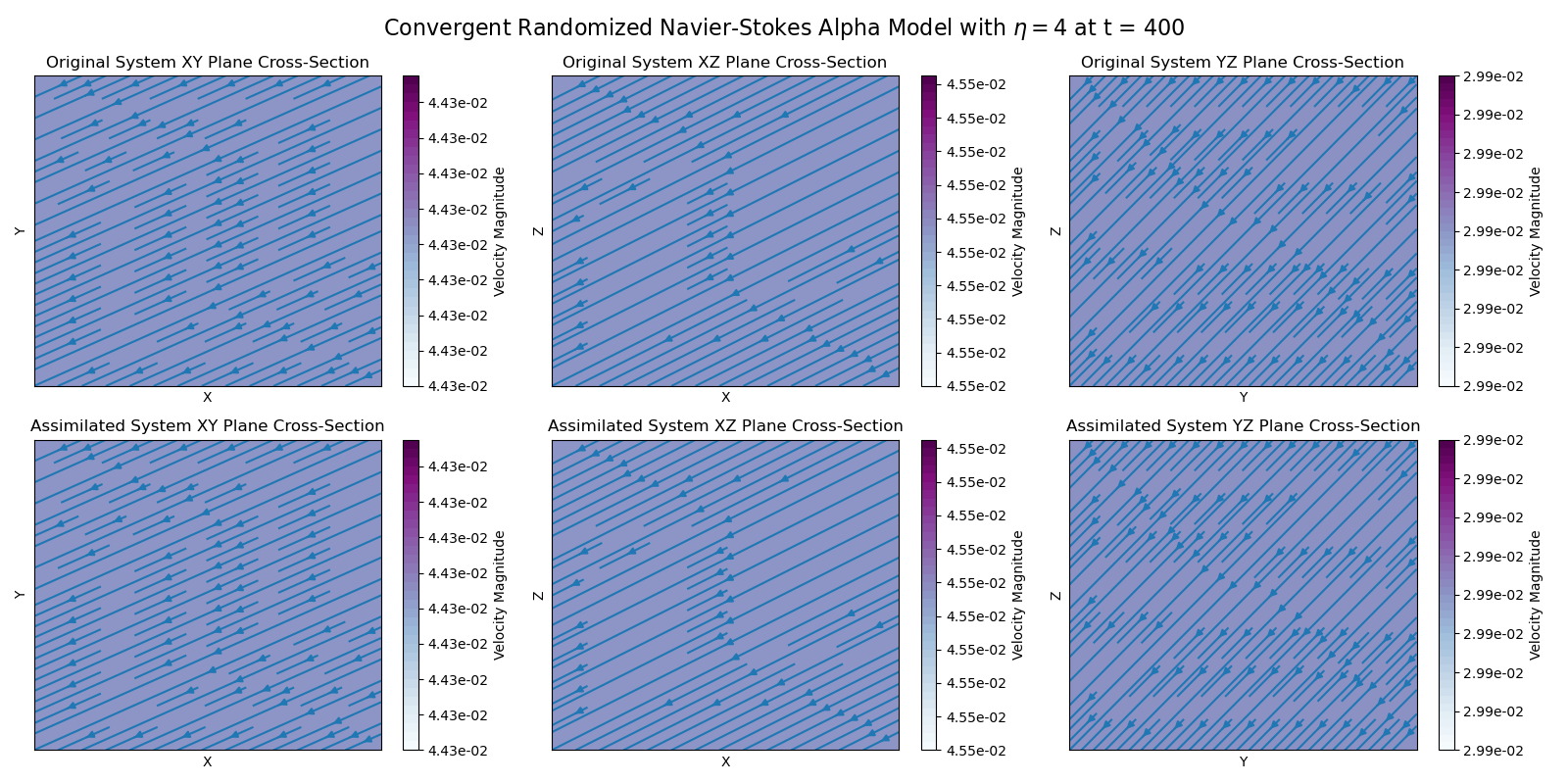}
    \caption{Velocity contour of NS-$\alpha$ model with high $\eta$ value-with random initial conditions case at $t=400$.}
    \label{fig21}
\end{figure}
\begin{figure}
\centering
    \includegraphics[totalheight=7cm]{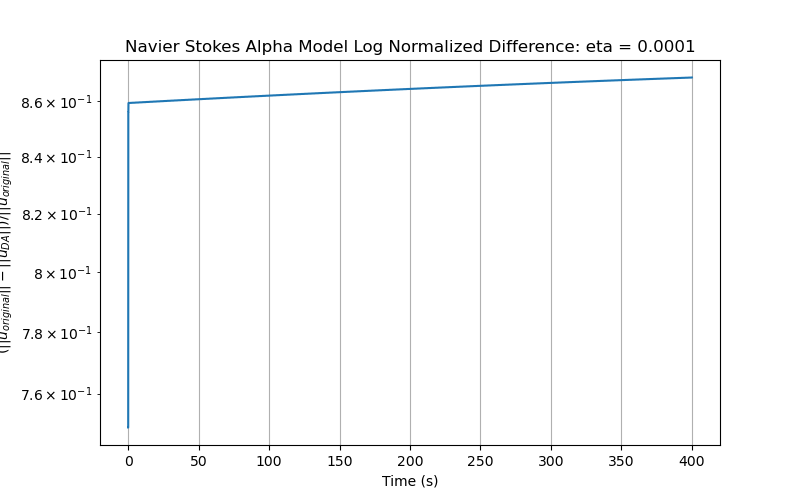}
    \caption{Error plot of NS-$\alpha$ model with high $\eta$ value-with random initial conditions case.}
    \label{fig22}
\end{figure}
\begin{figure}
\centering
    \includegraphics[totalheight=7cm]{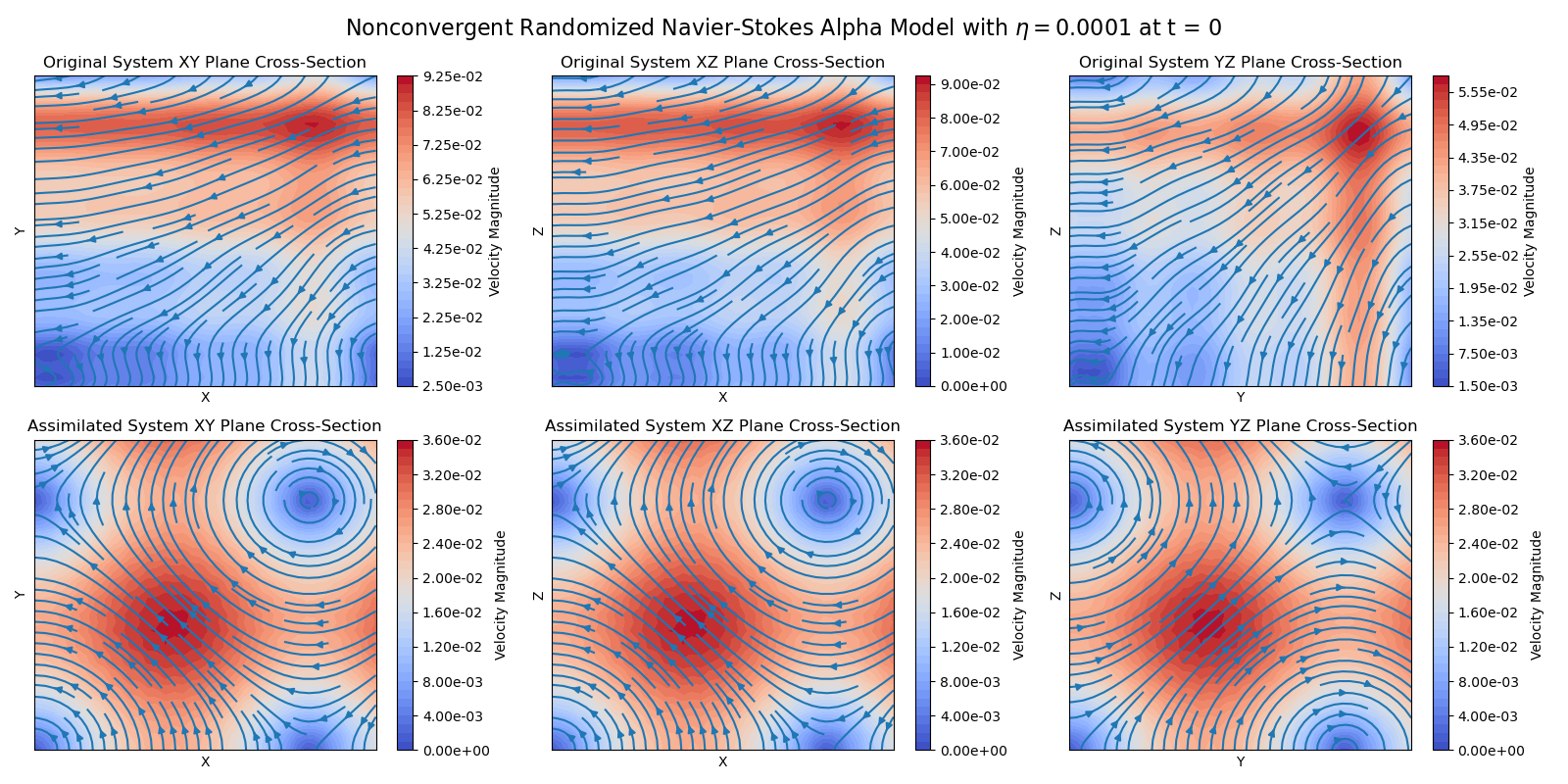}
    \caption{Velocity contour of NS-$\alpha$ model with high $\eta$ value-with random initial conditions case at $t=0$.}
    \label{fig23}
\end{figure}
\begin{figure}
\centering
    \includegraphics[totalheight=7cm]{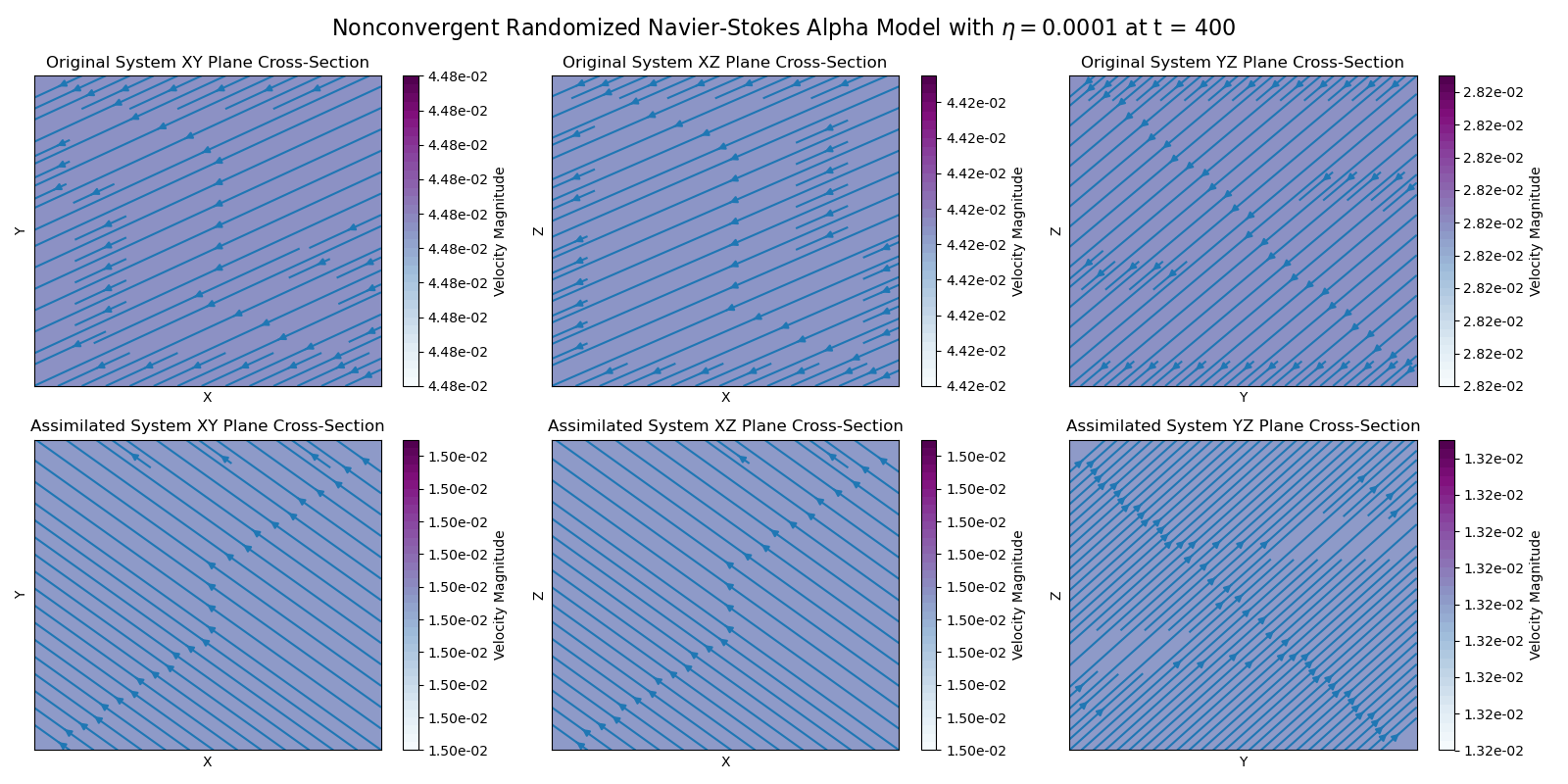}
    \caption{Velocity contour of NS-$\alpha$ model with high $\eta$ value-with random initial conditions case at $t=400$.}
    \label{fig24}
\end{figure}

\subsection{Discussions on the numerical computations} 
In the simulations, selecting parameters that satisfy those hypotheses while also demonstrating relevant behavior is a challenging task. 
In our selection, these conditions were carefully designed to account for several factors. Most importantly, the conditions between the assimilated and original systems need to differ enough so we could observe the convergence behavior. If the conditions were identical or very similar, the models would appear to converge even with very low eta values. Moreover, we choose such initial conditions so we can induce turbulence and differentiate the values of alpha and beta between the two systems. 


 The projection operator plays a crucial role in the nudging term used for implementing data assimilation in the turbulence models. In our numerical simulations, we have used the spectral methods and employ a Fourier projection to treat the nudging term. In spectral methods, differential equations are solved by representing the solution as a sum of basis functions, which transforms the differential equations into a system of algebraic equations in the frequency domain. The projection operator is applied to the nudging term, setting the Fourier coefficients of high modes to be zero.

During our numerical tests, we observed that convergence can still occur even when the conditions on $\eta$ and $h$ are not fully satisfied. For example, we found instances where convergence occurred despite $h$ being large or $\eta$ being small. The theoretical results in this paper are designed to guarantee success in worst-case scenarios, ensuring that the models will converge if the constraints are met, regardless of the initial conditions. However, since this paper does not focus on numerical analysis, we have included graphs of simulation that support our theoretical findings which will not necessarily diverge if the constraints are not met.
 
\section{Conclusions and Discussions}
In this work, we have studied the parameter estimation on the data assimilation algorithm for the simplified Bardina model and Navier-Stokes-$\alpha$ model. Our approach involves creating an approximate solution for the simplified Bardina model using an interpolant operator derived from observational data of the system. We have presented the long-time error between an approximate solution of the data assimilated system of both mentioned models to their real solutions. Different than literature studies, our approximate solution is with the length scale parameter $\alpha>0$ being considered unknown. We have demonstrated that, under appropriate conditions, the approximate solution converges to the real solution, with the exception of an error influenced by factors such as viscosity, the forcing term, and estimates in the $L^2$-norm of the real solution and its derivatives. Additionally, the error is evidently affected by the difference between the real and approximating parameters. We have also provided numerical simulations to support our theoretical findings.

In future, we plan to extend our studies on more turbulence models. We would also like to implement the parameter recovery algorithm for these turbulence models. In particular, we will design an algorithm to recover the unknown parameter $\alpha$.

\section{Appendix}
Proof of Lemma \ref{Gronwall}:
\begin{proof} Multiplying \eqref{ineq1} by $e^{Ct}$, and integrating over $[t_0,t]$, we get
$$\xi(t)\leq e^{-C(t-t_0)}\xi(t_0)+\ds\int_{t_0}^{t}e^{C(s-t)}\beta(s)\,ds$$

Now, fix $T>t_0$ and let $k_0\in\mathbb{N}$ such that
$$(k_0-1)T\leq t\leq k_0T.$$
Then
\begin{eqnarray}
    \xi(t)  \leq & \hspace{-2cm} e^{-C(t-t_0)}\xi(t_0)+\ds\sum_{k=1}^{k_0}\ds\int_{t_o+(k-1)T}^{t_0+kT}e^{C(s-t)}\beta(s)\,ds\nonumber\\
  \leq  & \hspace{-1cm} \nonumber e^{-C(t-t_0)}\xi(t_0)+\ds\sum_{k=1}^{k_0}\ds\int_{t_o+(k-1)T}^{t_0+kT}e^{C(kT-(k_0-1)T)}\beta(s)\,ds\nonumber\\
  \leq  & \nonumber e^{-C(t-t_0)}\xi(t_0)+\ds\sum_{k=1}^{k_0}e^{C(k-k_0+1)T}\sup_{k=1,...k_0}\ds\int_{t_0+(k-1)T}^{t_0+kT}\beta(s)ds. \nonumber
\end{eqnarray}
Therefore

\begin{eqnarray*}
    \xi(t)  \leq & \hspace{-2cm} e^{-C(t-t_0)}\xi(t_0)+Me^{C(1-k_0)T}\ds\sum_{k=1}^{k_0}e^{CkT} \nonumber\\
      \leq & e^{-C(t-t_0)}\xi(t_0)+Me^{C(1-k_0)T}\left(\ds\frac{e^{CT}-e^{C(k_0+1)T}}{1-e^{CT}}\right)\nonumber\\
      \leq & \hspace{-1.5cm} e^{-C(t-t_0)}\xi(t_0)+Me^{2CT}\left(\ds\frac{1-e^{-Ck_0T}}{e^{CT}-1}\right)\nonumber\\
     \leq & \hspace{-3cm} e^{-C(t-t_0)}\xi(t_0)+M\left(\ds\frac{e^{2CT}}{e^{CT}-1}\right).
\end{eqnarray*}
This brings us to equation (\ref{int211}).
\end{proof}

\textbf{Acknowledgment}
\vspace{3pt}

Samuel Little and Jing Tian's work is partially supported by the NSF LEAPS-MPS Grant $\#2316894$.

The authors would like to thank Yu Cao, Joshua Hudson, Michael Jolly, Leo Rebholz, and Jared Whitehead for their assistance and valuable suggestions regarding the numerical computations.



\printbibliography

\end{document}